\documentclass[12pt,reqno]{amsart}

\usepackage{enumerate}
\usepackage{graphicx}
\usepackage{epsfig}

\usepackage{xcolor}

\usepackage{amssymb,amsmath,amsthm,amsfonts}

\usepackage[letterpaper, margin=1in]{geometry} 

\usepackage {latexsym}

\usepackage{bbm}

\allowdisplaybreaks

\newcommand{\alphaS}{\alpha_{\mbox{\tiny{\textrm{S}}}}}

\newcommand{\nc}{\newcommand}

\nc{\sV}{\mathbf{s}}

\nc{\AV}{\mathbf{A}}
\nc{\BV}{\mathbf{B}}
\nc{\DV}{\mathbf{D}}
\nc{\EV}{\mathbf{E}}

\nc{\nullV}{\boldsymbol{0}}

\nc{\ellV}{\vec{\ell}}

\nc{\Cset}{\mathbb{C}}
\nc{\Nset}{\mathbb{N}}
\nc{\Rset}{\mathbb{R}}
\nc{\RR}{\mathbb{R}}
\nc{\Sset}{\mathbb{S}}
\nc{\Zset}{\mathbb{Z}}

\nc{\mEL}{m_{\mathrm{e}}}
\nc{\mPR}{m_{\mathrm{p}}}
\nc{\mZ}{m_{\mathrm{z}}}

\nc{\les}{\lesssim}
\nc{\nit}{\noindent}
\nc{\nn}{\nonumber}
\nc{\D}{\partial}
\nc{\diff}[2]{\frac{d #1}{d #2}}
\nc{\diffn}[3]{\frac{d^{#3} #1}{d {#2}^{#3}}}
\nc{\pdiff}[2]{\frac{\partial #1}{\partial #2}}
\nc{\pdiffn}[3]{\frac{\partial^{#3} #1}{\partial{#2}^{#3}}}
\nc{\abs}[1] {\lvert #1 \rvert}
\nc{\cAc}{{\cal A}_c}
\nc{\cE}{{\mathcal E}}
\nc{\cF}{{\cal F}}
\nc{\cP}{{\cal P}}
\nc{\cV}{{\cal V}}
\nc{\cQ}{{\cal Q}}
\nc{\cGin}{{\cal G}_{\rm in}}
\nc{\cGout}{{\cal G}_{\rm out}}
\nc{\cO}{{\cal O}}
\nc{\Lav}{{\cal L}_{\rm av}}
\nc{\cL}{{\cal L}}
\nc{\cB}{{\cal B}}
\nc{\cZ}{{\cal Z}}
\nc{\cR}{{\cal R}}
\nc{\cT}{{\cal T}}
\nc{\cY}{{\cal Y}}
\nc{\cX}{{\cal X}}
\nc{\cXT}{{{\cal X}(T)}}
\nc{\cBT}{{{\cal B}(T)}}
\nc{\vD}{{\vec \mathcal{D}}}
\nc{\efield}{\mathcal{E}}
\nc{\mE}{\mathcal{E}}
\nc{\vE}{{\vec \efield}}
\nc{\vB}{{\vec \mathcal{B}}}
\nc{\vH}{{\vec \mathcal{H}}}
\nc{\F}{  \mathcal{F} }
\nc{\ty}{{\tilde y}}
\nc{\tu}{{\tilde u}}
\nc{\tV}{{\tilde V}}
\nc{\Pc}{{\bf P_c}}
\nc{\bx}{{\bf x}}
\nc{\bX}{{\bf X}}
\nc{\bXYZ}{{\bf XYZ}}
\nc{\bY}{{\bf Y}}
\nc{\bF}{{\bf F}}
\nc{\bS}{{\bf S}}
\nc{\dV}{{\delta V}}
\nc{\dE}{{\delta E}}
\nc{\TT}{{\Theta}}
\nc{\dPsi}{{\delta\Psi}}
\nc{\order}{{\cal O}}
\nc{\Rout}{R_{\rm out}}
\nc{\eplus}{e_+}
\nc{\eminus}{e_-}
\nc{\epm}{e_\pm}

\nc{\vnabla}{{\vec\nabla}}
\nc{\G}{\Gamma}
\nc{\w}{\omega}
\nc{\mh}{h}
\nc{\mg}{g}
\nc{\vphi}{\varphi}
\nc{\tlambda}{\tilde\lambda}
\nc{\be}{\begin{equation}}
\nc{\ee}{\end{equation}}
\nc{\ba}{\begin{eqnarray}}
\nc{\ea}{\end{eqnarray}}

\nc{\g}{\gamma}
\nc{\ol}{\overline}
\nc{\n}{\nu}

\newtheorem{thm}{Theorem}[section]
\newtheorem{lemm}[thm]{Lemma}
\newtheorem{prop}[thm]{Proposition}
\newtheorem{cor}[thm]{Corollary}

\newtheorem{rmk}[thm]{Remark}

\newtheorem{asmps}[thm]{Assumptions}

\nc{\pr}{\partial_r}
\nc{\pt}{\partial_t}
\nc{\pT}{\partial_T}
\nc{\pz}{\partial_z}

\nc{\la}{\langle}
\nc{\ra}{\rangle}
\nc{\infint}{\int_{-\infty}^{\infty}}
\nc{\halfwidth}{6.5cm}
\nc{\figwidth}{10cm}
\newcommand{\f}{\frac}

\nc{\nlayers}{L} \nc{\nsectors}{M}
\nc{\indicator}{\mathbf{1}}
\nc{\Rhole}{R_{\rm hole}}
\nc{\Rring}{R_{\rm ring}}
\nc{\neff}{n_{\rm eff}}
\nc{\Frem}{F_{\rm rem}}
\nc{\R}{\mathbb R}
\nc{\C}{\mathbb C}
\nc{\Z}{\mathbb Z}
\nc{\DD}{\Delta}
\nc{\cD}{\mathcal D}
\nc{\lnorm}{\left\|}
\nc{\rnorm}{\right\|}
\nc{\rnormp}{\right\|_{\ell^{p,\eps}}}
\nc{\rar}{\rightarrow}
\nc{\mR}{\mathcal R}
\nc{\oo}{\"o}   

\nc{\os}{\overset{o}}

\nc{\eps}{\epsilon}
\nc{\veps}{\varepsilon}

\nc{\dd}{{\mathrm{d}}}

\nc{\mNULL}{{C_\alpha^{}}}
\nc{\phiNULL}{{C_\beta^{\prime}}}

\nc{\beq}{\begin{equation}}
\nc{\eeq}{\end{equation}}

\begin{document}

\title[General-relativistic hydrogen]{On general-relativistic hydrogen and hydrogenic ions}

\author[Kiessling, Tahvildar-Zadeh, Toprak]{Michael K.-H. Kiessling, A. Shadi Tahvildar-Zadeh, Ebru Toprak}

 \address{Department of Mathematics \\
Rutgers University \\
Piscataway, NJ 08854, U.S.A.}
\email{miki@math.rutgers.edu}
 \address{Department of Mathematics \\
Rutgers University \\
Piscataway, NJ 08854, U.S.A.}
\email{shadit@math.rutgers.edu}
\address{Department of Mathematics \\
Rutgers University \\
Piscataway, NJ 08854, U.S.A.}
\email{et400@math.rutgers.edu}

\begin{abstract}
 This paper studies how the static non-linear electromagnetic-vacuum spacetime of a point nucleus 
with negative bare mass affects the self-adjointness of the {general-relativistic} Dirac Hamiltonian for a test electron, 
without and with an anomalous magnetic moment.
 The study interpolates between the previously studied extreme cases of a test electron in (a) the Reissner--Weyl--Nordstr\"om
spacetime ({Maxwell's electromagnetic vacuum}), which sports a very strong curvature singularity with negative infinite bare mass,
and (b) the Hoffmann spacetime (Born or Born--Infeld's electromagnetic vacuum) with vanishing bare mass, 
which features the mildest possible curvature singularity.
 The main conclusion reached is: on electrostatic spacetimes of a point nucleus with a strictly negative bare mass
(which may be $-\infty$) essential self-adjointness fails unless the radial electric field diverges sufficiently fast at 
the nucleus \emph{and} the anomalous magnetic moment of the electron is taken into account. 
 Thus on the Hoffmann spacetime with (strictly) negative bare mass the Dirac Hamiltonian of a test electron, with or without 
anomalous magnetic moment, is not essentially self-adjoint.
 All these operators have self-adjoint extensions, though, with the usual essential spectrum $(-\infty,-\mEL c^2]\cup[\mEL c^2,\infty)$
and an infinite discrete spectrum located in the gap $(-\mEL c^2,\mEL c^2)$.\vspace{-30pt}
\end{abstract}


\date{Corrected version of Aug. 28, 2020; Printed: \today\\ \copyright (2020): The authors. Reproduction of this article for non-commercial purposes is 
permitted.}

\maketitle

\section{Introduction}

\subsection{State of Affairs}

 In non-relativistic physics, whether Newtonian mechanics or quantum mechanics, 
the gravitational and the electrical attraction between a point electron and a point proton
obey the same mathematical {force} law and only their coupling strengths differ --- though by a lot:
 If (in Gaussian units) $e$ denotes the elementary charge, $\mEL$ the empirical mass of the electron and 
$\mPR$ the one of the proton, and $G$ is Newton's constant of universal gravitation, then
$\frac{G\mPR\mEL}{e^2}=:\gamma_{pe}\approx 4.5\cdot 10^{-40}$. 
 Thus, in such theories gravity is an extremely weak pair interaction between electron and proton, 
compared to electricity, indeed so weak
that it is hard to imagine how any experimental study of the hydrogen atom's spectrum could possibly reveal its effects
--- assuming that non-relativistic quantum mechanics predicts the effect accurately enough for all practical purposes.
 Explicitly, the bound state spectrum of hydrogen in non-relativistic QM is readily obtained from the familiar Bohr 
formula through the replacement $e^2\mapsto e^2 + G\mEL\mPR$, viz.
\begin{equation}\label{BOHRspec}
\frac{1}{\mEL c^2} E_n^{\mbox{\tiny{Bohr}}}(Z,N;\gamma_{pe})\Big|_{Z=1,N=0}\Big.
= -\frac12 \alphaS^2 \frac{(1 + \gamma_{pe})^2}{1+\eps\ } \frac{1}{n^2},\qquad n\in\Nset;
\end{equation}
here, $\eps:=\mEL/\mPR\approx 1/1836$, and $\alphaS:={e^2}/{\hbar c}\approx 1/137.036$ is Sommerfeld's fine
structure constant, where $\hbar$ is the Planck constant divided by $2\pi$, and $c$ the speed of light.
 Note that $\frac{\mEL}{1+\eps} = \frac{\mEL\mPR}{\mEL +\mPR}$ is the reduced mass of the electron-proton system.
 We see that each Bohr level is lowered by a factor $\approx 1+ 10^{-39}$ compared
to the result for purely electrical Coulomb interaction, $ E_n^{\mbox{\tiny{Bohr}}}(1,0;0)$.
 This effect is almost 30 orders of magnitude smaller than the best spectral resolution achieved today.

 The non-relativistic gravitational effects on the spectrum of a hydrogenic ion are
slightly more pronounced, though not in any significant way. 
 To obtain the Bohr spectrum of a hydrogenic ion, replace the proton charge $e\mapsto Ze$ 
and the proton mass\footnote{Here, {$A(Z,N)\approx Z+N$, roughly} 
the number of nucleons in a nucleus.\vspace{-5pt}}  
{$\mPR\mapsto A(Z,N)\mPR$,  with $Z\in\Nset$, $N\in\{0,1,2,...\}$, and $A(Z,N)\geq Z$
(N.B.: $Z\leq 118$ and $N<200$ in the currently known chart of the nuclids, and
$Z\leq A(Z,N) \leq 3Z$ for the known long-lived nuclei; for hydrogen: $A(1,0)=1$).}
 Thus (\ref{BOHRspec}) is the $Z=1$ \&\ $N=0$ special case of
\begin{equation}\label{BOHRspecZ}
\frac{1}{\mEL c^2} E_n^{\mbox{\tiny{Bohr}}}(Z,N;\gamma_{pe}) 
=  -\frac12 \alphaS^2
\frac{\big(Z + \gamma_{pe} A(Z,N)\big)^2}{1+\eps/A(Z,N)} \frac{1}{n^2},\qquad n\in\Nset.
\end{equation}
 So $E_n^{\mbox{\tiny{Bohr}}}(Z,N;\gamma_{pe})$ differs from $E_n^{\mbox{\tiny{Bohr}}}(Z,N;0)$
by not more than $3\cdot 10^{-39} E_n^{\mbox{\tiny{Bohr}}}(Z,N;0)$.

\begin{rmk}
We recall that the Bohr model yields the same energy spectrum for hydrogenic atoms / ions as does the 
Schr\"odinger Hamiltonian.
 We also recall that the spectrum in the Born--Oppen\-heimer approximation (the electron is treated
as a \emph{test particle} in the Coulomb field of a fixed nucleus) 
is recovered by letting $\mPR\to\infty$, equivalently $\eps\to 0$ in (\ref{BOHRspec}), (\ref{BOHRspecZ}).
\end{rmk}

 The anticipated tininess of the gravitational effect in the hydrogen spectrum is also the reason why Sommerfeld \cite{Somm} 
did not generalize his special-relativistic calculations of the hydrogen spectrum (in the Born--Oppenheimer approximation) 
to the freshly created general-relativistic setting.
 In fact, Sommerfeld had consulted with Einstein prior to publication of \cite{Somm} whether it would be advisable to include the 
general-relativistic effects, but Einstein advised against it \cite{EtoS}, stating that the quantitative results would essentially
 agree with Sommerfeld's fine structure formula obtained by invoking only special relativity (for the kinetic energy of the 
electron), and only Coulomb electricity for the interaction between electron and proton; see also \cite{SommBOOK}.
 
 However, relativistic electricity (read: electromagnetism) and gravity (read: spacetime curvature) are no longer 
mathematically identical structures, and so their relative contributions to the atomic spectra cannot obviously be 
estimated merely in terms of a comparison of their coupling constants. 
 Sure enough, not long after Sommerfeld published his work on the relativistic hydrogen fine structure he was criticized 
by Wereide \cite{W} for not having mathematically demonstrated that general-relativistic effects were indeed so 
tiny as to be negligible.
 Eventually, Vallarta in his MIT Ph.D. thesis (the main results are published in \cite{V}) supplied mathematically definitive
estimates of the general-relativistic effects in the Bohr--Sommerfeld-type spectrum of hydrogen.
 Vallarta considered a test electron in the Reissner--Weyl--Nordstr\"om (RWN) spacetime with a naked timelike singularity, 
equipped with the electric charge of the proton, and the ADM mass equated with the empirical 
proton mass, and applied the Bohr--Sommerfeld quantization rules to the bound electron orbits.
 He concluded that the relativistic gravitational effects were immeasurably tiny, and so he did not even bother to actually 
compute their corrections to Sommerfeld's fine structure spectrum, although he could have done so with the help of perturbation theory.
 Such computations were done recently, for circular orbits, by Dreifus in her honors thesis at Rutgers \cite{Dr}.

 Even though Vallarta's estimates and Dreifus' perturbative computations have produced quantitatively tiny general-relativistic 
corrections to the special-relativistic Sommerfeld fine-structure spectrum of hydrogen using Bohr--Sommerfeld-type quantization,
it would be quite a mistake to {now} conclude from this that general relativity would always manifest itself only in 
form of a tiny perturbation of special-relativistic atomic spectral results. 
 As emphasized already, the general theory of relativity reveals that gravity is not a weaker attractive `clone' of electromagnetism, 
but a completely different `force of nature.' 
 There is little doubt nowadays that general relativity correctly predicts that nature is capable of forming black holes which 
can swallow unlimited amounts of matter as long as supplies will last. 
 Intuitively, therefore, one would be inclined to {suspect} that general relativity {should have} 
a destabilizing effect in the theory of large-$Z$ atoms. 
 At the very least one might expect a worsening of the spectral `large-$Z$ catastrophe' in the special-relativistic ($G=0$) 
Bohr--Sommerfeld theory of hydrogenic ions, where it occurs when the nuclear charge number $Z$ exceeds $1/\alphaS$
{and the bottom drops out from under the energy functional {because the electrical attraction
overpowers the angular momentum barrier of the circular motion}; see our Appendix~{A}}. 

 Curiously, `switching on {relativistic} gravity' instead \emph{removes} this `large-$Z$ catastrophe' of
 the special-relativistic Bohr--Sommerfeld-type model for {hydrogenic ions}.
 Namely, in our Appendix~{A} we show that for each $Z\in\Nset$ there is a unique Bohr--Sommerfeld-type spectrum
of the general-relativistic hydrogenic ion obtained from a minimum-energy variational principle.
 By contrast, in the special-relativistic Bohr--Sommerfeld-type model of a hydrogenic ion, 
the pertinent minimum-energy variational principle has no lower bound when $Z>1/\alphaS$. 

 The just mentioned catastrophe at $Z = 1/\alphaS$ in the special-relativistic Bohr--Sommerfeld model of hydrogenic ions
has a counterpart in the spectral theory of the special-relativistic Dirac Hamiltonian for a hydrogenic atom/ion \cite{Rose,Thaller,GreinerETal},
where there is also an earlier catastrophe at $Z = \sqrt{3}/2\alphaS$.
 We recall that this Dirac operator is essentially self-adjoint 
only if\footnote{Allowing $Z\in \Rset_+$ essential self-adjointness holds for $Z\leq  {\sqrt{3}}/{2\alphaS}$.}
$Z\leq 118$, yet it has a (unique) analytical extension (to $Z\in\Cset$) which is self-adjoint also when 
$Z\in\{119,...,137\}$, but the analytical 
extension is no longer self-adjoint 
when\footnote{Allowing $Z\in \Rset_+$ the analytical extension is self-adjoint for $Z\leq 1/\alphaS\approx 137.036$.\vspace{-5pt}}
$Z>137$; cf. \cite{WeiA,Narnhofer,Thaller,EsLo}. 
 As pointed out by Narnhofer \cite{Narnhofer}, for $Z {\geq 119}$ 
the deficiency indices of the Dirac operator restricted to a fixed angular momentum subspace
are $(1,1)$, so there always exist self-adjoint extensions of the formal Dirac operator for hydrogenic ions, 
but for $Z>137$ it is not clear which one, if any, is physically distinguished. 

\begin{rmk}
We recall that Sommerfeld's fine-structure formula for the energy spectrum of hydrogen agrees with the 
hydrogen spectrum obtained with Dirac's special-relativistic wave equation for an electron in the Coulomb
field of a fixed proton.
 The assignment of angular momentum quantum numbers in Sommerfeld's calculations of
course does not agree with the spectral formula of the special-relativistic Dirac Hamiltonian, for Sommerfeld
did not incorporate any form of electron spin. 
 The subtle reason for this remarkable coincidence of the Sommerfeld and Dirac energy spectra for hydrogen is nicely explained
in \cite{Keppeler}.
\end{rmk}

 Since the Bohr--Sommerfeld theory of the spectra of hydrogenic ions exactly captures their quantum-mechanical energy spectra 
in both the non-relativistic (Schr\"odinger) setting and in the special-relativistic (Dirac) setting for $Z\alphaS \leq 1$,
and since general relativity has a regularizing effect on the Bohr--Sommerfeld theory, 
\emph{at this point} it certainly would seem reasonable to expect that general relativity will have a regularizing 
effect also in the Dirac theory of hydrogenic spectra. 
 However, the opposite is true!

 Namely, as discovered by Cohen and Powers \cite{CP}, the general-relativistic Dirac Hamiltonian \cite{Erwin},
\cite{BC} for hydrogen differs dramatically from the familiar special-relativistic Dirac Hamiltonian for hydrogen.
 More precisely, like Vallarta, so also Cohen and Powers modelled general relativistic hydrogen as consisting of
a \emph{test electron} in the static Reissner--Weyl--Nordstr\"om spacetime of a fixed point proton. 
 While the special-relativistic Dirac Hamiltonian for hydrogenic ions is essentially self-adjoint 
(on the domain $C^\infty_c(\Rset^3\backslash\{0\})^4$) for all $Z\leq 118$ \cite{Narnhofer,Thaller}, 
Cohen and Powers discovered that the general-relativistic hydrogen Hamiltonian is not --- it has uncountably many 
self-adjoint extensions; the same conclusion holds for all $Z>0$ in the hydrogenic problem.
 By Stone's theorem, each one of these is the generator of a different unitary evolution, so
the question becomes: which one (if any) is the physically correct self-adjoint extension?
 If general-relativistic effects in the spectrum are immeasurably small, then 
empirical spectral data for hydrogen will not help to find the answer; cf. \cite{Parker}.

 The essential spectrum of any self-adjoint extension of the 
general-relativistic Dirac Hamiltonian of an electron in the RWN spacetime of a nucleus
was determined in \cite{Belgiorno}, and in Appendix C of \cite{BMB} Belgiorno et al. showed that there 
are infinitely many bound states in the gap $(-\mEL c^2, \mEL c^2)$ of the essential spectrum. 
 As far as we are aware, it is not known whether any of these point spectra converges to the Sommerfeld fine 
structure spectrum when $G\searrow 0$.
 
\begin{rmk}
 At the end of the day, the findings of Cohen and Powers vindicate the earlier expressed intuition
that general relativity might worsen the spectral `large-$Z$ catastrophe' in the special-relativistic ($G=0$) 
treatment of hydrogenic ions, except that this turns out to be true for the Dirac theory of the energy spectra, 
not for the Bohr--Sommerfeld theory. 
 Specifically, the first `large-$Z$ catastrophe' in the special-relativistic Dirac theory of hydrogenic ions (i.e.
the loss of essential self-adjointness when the nuclear charge number $Z$ exceeds the critical value $\sqrt{3}/2\alphaS$)
is worsened, with the critical $Z$-value reduced to $0$ if $G>0$.
\end{rmk}

Now, the Dirac Hamiltonian for a point electron in an externally generated magnetostatic induction field 
$\BV(\sV) = \nabla\times\AV(\sV)$ automatically endows the electron with a magnetic moment of magnitude 
$\mu_{\mbox{\tiny{Bohr}}} = \frac{1}{4\pi}\frac{h e}{\mEL c}$ and a $g$ factor of $2$. 
 Empirically, the electron does seem to have a magnetic moment which differs slightly from the Bohr magneton, though,
and the difference is known as its anomalous magnetic moment $\mu_a$.
 Using perturbative QED it has been computed in terms of a truncated power series in powers of $\alphaS$ (and $\log \alphaS$). 
 Interestingly, the leading order term in the expansion of the anomalous magnetic moment $\mu_a$ is independent of $\hbar$ and reads 
$\mu_{\mbox{\tiny{class}}} = \frac{1}{4\pi}\frac{e^3}{\mEL c^2}$, which we call \emph{the classical magnetic moment of the
electron}.
 It already gives a very accurate value for the anomalous magnetic moment of the electron.

  It has been known for a long time that the addition of an anomalous magnetic moment operator to the Dirac Hamiltonian of a test
electron with purely electrostatic interactions removes both of the spectral `large-$Z$ catastrophes,' in the sense that it produces 
an essentially self-adjoint Hamiltonian for the electron of any hydrogenic ion \cite{Beh,GST}, independently of the strength of the 
non-vanishing anomalous magnetic moment. 
 More recently Belgiorno, Martellini, and Baldicchi \cite{BMB} showed that the Dirac operator with anomalous magnetic moment is essentially 
self-adjoint in the naked {RWN} geometry only if $|\mu_a| \geq \frac{ 3}{2} \frac{ \sqrt{G} \hbar}{c}$.
 Since 
\be 
 \frac{ \sqrt{G} \hbar}{c} = \sqrt{\frac{G\mEL^2}{e^2}} {\frac{ \hbar c}{ e^2}  \frac{e^3}{\mEL c^2} }
=
  \sqrt{\gamma_{pe}\eps} \frac{1}{ \alpha_s} 4 \pi \mu_{\text{class}} ,
\ee
the requirement
$|\mu_a| \geq \frac{ 3}{2} \frac{ \sqrt{G} \hbar}{c}$ corresponds to {$|\mu_a| \gtrsim 1.3 \cdot 10^{-18} \mu_{\text{class}} $,
which is manifestly satisfied by the empirical value 
$|\mu_a| \approx \mu_{\text{class}}$} of the electron's anomalous magnetic moment.

 While general relativity therefore does not have a catastrophic effect in the spectral theory of physical hydrogenic ions thanks to 
the sufficiently large empirical value of the anomalous magnetic moment of the electron, it still would have a catastrophic 
effect if the empirical value were much smaller. 
 In this mathematical sense `switching on general-relativistic gravity' is generally \emph{not} a {harmless} weak perturbation \cite{Kato} 
of the  essentially self-adjoint special-relativistic Dirac operator with Coulomb electricity and anomalous magnetic moment!

 The RWN spacetime of a proton has a number of suspicious features, though {\cite{Weyl}}.
 In particular, it has a very strong curvature singularity at its center. 
 Also, its electrostatic field energy is infinite, but it has a finite positive ADM mass (which is identified with the
mass $\mPR$ of the proton). 
 This suggests that the RWN spacetime singularity sports a negative infinite bare mass; indeed this can be computed
using the Hawking mass formula for the mass in the immediate vicinity of the naked point singularity of the RWN spacetime of
the proton.

 The origin of the divergent field energy is long known: the same divergence occurs in flat spacetime, namely point charges in
Lorentz electrodynamics have an infinite self-field energy. 
 And since therefore the energy-momentum-stress tensor of the Maxwell--Lorentz fields with a point charge source is not locally
integrable over any vicinity of the point charge, coupling it via Einstein's equations to the Ricci curvature of spacetime
will inevitably cause very strong spacetime singularities, no matter how tiny the gravitational coupling constant is.
 This suggests that the problems may go away if one works with an electromagnetic field theory of non-linear electromagnetic 
vacua which give rise to an energy-momentum-stress tensor of the electrostatic 
field with a point charge source which \emph{is} globally integrable.

 Prominent examples are the Born and the Born--Infeld vacuum laws; we recall that they coincide in the electrostatic limit.
 As Born found \cite{BornA}, the electrostatic potential field of a point charge in a Born(--Infeld) vacuum 
is bounded and Lipschitz continuous.
 It now follows from quite general results about the spherically symmetric special-relativistic Dirac Hamiltonian \cite{KSWW,Thaller}
that for a test electron in the electrostatic Born field of a point nucleus of charge $Ze$ the Hamiltonian
is essentially self-adjoint for $Z\in\Rset$.
 
 The question thus becomes whether the elimination of the infinite electrostatic self-field energy problem with the help of
some non-linear vacuum law such as the Born(--Infeld) law suffices to guarantee an essentially self-adjoint Dirac Hamiltonian 
for a test electron {also} in the general-relativistic spacetime of a point nucleus.
 Since the electrostatic spacetime of a stable nucleus should have an ADM mass identical to $A(Z,N)\mPR$, 
if the energy-momentum-stress tensor of the electrostatic field with a point charge source 
is integrable then also the bare mass of the central singularity has to be finite. 
 To avoid a black hole, it has to be non-positive.
 However, a non-positive bare mass of the central singularity alone is not sufficient to avoid a black hole; 
further conditions need to be met, but they can.

 Balasubramanian in his Ph.D. thesis \cite{Moulik} showed that the Dirac Hamiltonian for an electron in the Hoffmann spacetime 
\cite{Hoffmann} of a point nucleus with \emph{zero bare mass} is essentially self-adjoint for all $Z\in\Nset$.
 He actually showed it for a larger class of similar black-hole-free electrostatic spacetimes \cite{TZ}, all 
having \emph{zero bare mass}.

\subsection{Terra incognita}
 The works \cite{CP}, \cite{BMB}, and \cite{Moulik} concern
two `opposite' endpoints of a large multi-parameter family of black-hole-free electrostatic spacetimes which in a sense
interpolate between the two extreme cases.
 Thus their results have left open the question of what happens in black-hole-free electrostatic spacetimes with 
either integrable field energy-momentum-stress tensor, yet with a finite strictly negative bare mass at their center, or
with non-integrable field energy density function $\veps(r)$, and thus with a negative infinite bare mass at their center, 
yet with a radial mass function
\begin{equation}\label{mOFr}
m(r) := M_{\mbox{\tiny{ADM}}}- \tfrac{1}{c^2} \int_r^\infty \veps(s)4\pi s^2\dd{s}
\end{equation}
which diverges to $-\infty$ as $r\downarrow 0$ slower than the RWN mass function
\begin{equation}\label{mOFrRWN}
m_{\mbox{\tiny{RWN}}}(r) := M_{\mbox{\tiny{ADM}}} - \tfrac12\tfrac{Z^2e^2}{c^2}\tfrac1r ;
\end{equation}
recall that $M_{\mbox{\tiny{ADM}}} = A(Z,N)\mPR$ for a nucleus of charge $Ze$.
 Will the Dirac operator of a test electron in any of these spacetimes 
be more similar to the RWN case or to the Hoffmann case with vanishing bare mass?
 Or will there be a critical strictly negative borderline value of the bare mass where a switch-over happens?
 And what is the influence of the electron's anomalous magnetic moment?
 These questions do not have an obvious answer.

\begin{rmk}
 We remark that a strictly negative bare mass of the nucleus {seems} hard to avoid theoretically. 
 Recall that in non-perturbative renormalized QED, which needs an UV cutoff, the bare mass of the electron is 
strictly negative [and even $-\infty$ in perturbative QED], and similar conclusions are to be expected for
nuclei due to their electric charges.
\end{rmk}

\subsection{This Paper}
 In this paper we study the Dirac operator on a class of electrostatic spacetimes which includes those studied in \cite{TZ}
as well as the RWN spacetime with naked singularity.
 We show that whenever the bare mass of the central singularity of the electrostatic black-hole-free 
spacetime of a point nucleus is strictly negative, possibly negatively infinite,
 then the Dirac Hamiltonian for hydrogen / hydrogenic ions {without anomalous magnetic moment of the electron}
has uncountably many self-adjoint extensions. 
 Any of these self-adjoint extensions has purely absolutely continuous spectrum in $(-\infty,-\mEL c^2)\cup (\mEL c^2, +\infty)$, 
its closure being the essential spectrum, plus
a discrete spectrum with infinitely many eigenvalues located in the gap $(-\mEL c^2,\mEL c^2)$ of the essential spectrum.
 Which one of these, if any, is the physically correct one is an open question.
 This demonstrates that the local non-integrability of the field energy-momentum-stress tensor over any neighborhood of the point charge,
a feature of the RWN spacetime, is not the only source of trouble for the Dirac Hamiltonian of general-relativistic hydrogen.

We also address the question whether the addition of a sufficiently large anomalous magnetic moment operator to the 
Dirac Hamiltonian for such hydrogenic ions will result in an essentially self-adjoint operator in all cases studied here 
where essential self-adjointness fails without such an anomalous magnetic moment.
 In particular, in all spacetimes studied here the curvature singularity is milder than the one of the RWN spacetime, so that
 one might expect a lowered threshold value for the strength of the electron's anomalous magnetic moment.
 Interestingly, the situation is more complicated!

 Namely, while we find that there is a family of electromagnetic vacuum laws for which essential self-adjointness of the 
Dirac operator on the pertinent spacetime of a nucleus holds when the test electron exhibits any anomalous magnetic moment,
no matter how small, there also is another family --- which includes the {Born and}
Born--Infeld vacuum laws --- for which the addition of an anomalous magnetic moment operator of any strength, no matter how large, 
is not sufficient to obtain an essentially self-adjoint Dirac Hamiltonian.
 Explicitly, this means that the Dirac operator for a test electron in the Hoffmann spacetime of a nucleus with 
(inevitably finite) negative bare mass is not essentially self-adjoint, with or without the anomalous magnetic moment of the electron. 
\newpage

The rest of the paper is structured as follows:

In section \ref{sec:EST} we stipulate the class of electrostatic spacetimes considered in this paper;
with some technical details relegated to Appendix~B.

 In section \ref{sec:Dirac} we discuss the Dirac operator for a test electron in the type of electrostatic spacetime
defined in section \ref{sec:EST}. 
 The section is devided into two subsections, one devoted to test electrons without, and one to test electrons with
anomalous magnetic moment.

In section \ref{sec:SUMM} we offer an outlook on open questions to be addressed in some future work.

In Appendix~{A} we {explain the generally regularizing effect of 
general relativity in the Bohr--Sommerfeld type theory of quantized circular orbits.}

\section{Electrostatic spacetimes with {negative} bare mass {and no horizon}}\label{sec:EST}

 The electrostatic spacetimes discussed in this paper are equipped 
with an electromagnetic vacuum law derived from a Lagrangian density which is a function of
the two invariants of the Faraday field tensor $\boldsymbol{F}$.
 As shown {already} in \cite{TZ}, the spherically symmetric, static, asymptotically flat ones among them which
are topologically identical to `$\Rset^{1,3}$ minus a timelike line,' equivalently $\Rset\times (\Rset^3\setminus\{0\})$,
and covered by a single global chart of `spherical coordinates' 
$(t,r,\vartheta,\varphi)\in \Rset\times\Rset_+\times[0,\pi]\times[0,2\pi)$, have a metric given by the line element
\begin{equation}\label{dsSQR}
ds^2 = - f^2(r) c^2 \dd t^2+ \frac{1}{f^2(r)} \dd r^2 + r^2 \dd\Omega^2, \qquad  f^2(r) = 1 - 2\frac{G}{c^2}\frac{m(r)}{r} .
\end{equation}
 Here, $r$ is the so-called area radius of a spherical orbit; i.e., every point in the stipulated spacetime is an element of
a unique orbit under a Killing vector flow corresponding to the $SO(3)$ symmetry, and this orbit is a scaled copy of 
$\Sset^2$ with area $A =: 4\pi r^2$, \emph{defining} $r>0$. 
 Next, $\dd\Omega^2 = \dd\vartheta^2 +\sin^2\vartheta\dd\varphi^2$ is the line element on $\Sset^2$. 
 Moreover, 
\begin{equation}
 m(r)c^2 =  {M}c^2 - \cE(r)
\end{equation}
is the radial mass function, where $M$ is the ADM mass $M_{\mbox{\tiny{ADM}}}$ and  
$\cE(r)$ is the electrostatic field energy outside a ball of surface area $4\pi r^2$. 
 The field energy function $r\mapsto\cE(r)$ is strictly positive and monotone decreasing to $0$. 
 
 For the class of models studied in \cite{TZ} and here, $\cE(r)$ turns out to be
independent of $G$ and, hence, identical to the corresponding flat-space formula.
 Thus, for instance, in Maxwell--Lorentz electrodynamics, if $\sV\in\Rset^3$ is a point in flat space and $s:=|\sV|$, and 
$\EV(\sV)\in\Rset^3$ denotes the electric field strength vector at $\sV$ of a point nucleus located at $\nullV$, then 
$\cE(r) = \frac{1}{8\pi}\int_r^\infty |\EV(\sV)|^2\, 4\pi s^2\dd s$ with $|\EV(\sV)| = |\EV|(s) = {Ze}/{s^2}$.
\newpage

 Here are two well-known examples of such spacetimes.

 First, for the RWN spacetime of a nucleus of charge $Ze$, we have
\begin{equation}
\cE(r) =  \frac{1}{8\pi} \int_r^\infty \frac{Z^2e^2}{s^4} 4\pi s^2 \dd s = \frac12\frac{Z^2e^2}{r}.
\end{equation}
 Clearly, $m(r) \searrow -\infty$ as $r\searrow 0$, but we also want to have a spacetime without a black hole. 
 The RWN spacetime features a black hole if there is at least one value of $r>0$ for which $f^2(r) =0$. 
 Since $f^2(r)$ is a quadratic polynomial in $1/r$, its zeros are formally given by 
\begin{equation}
 r_\pm = \tfrac{GM}{c^2} \left( 1\pm \sqrt{1- \tfrac{Z^2e^2}{GM^2}}\right),
 \end{equation}
and this is real \emph{if and only if} $\tfrac{Z^2e^2}{GM^2}\leq 1$. 
 However, for the known nuclei
$M_{\mbox{\tiny{ADM}}} = A(Z,N)\mPR$, with $Z\leq A(Z,N) {\leq} 3Z$, and $\tfrac{e^2}{G\mPR^2}\approx 5\cdot 10^{36}$, 
so $f^2(r)$ is never zero and we are deep in the naked singularity sector of the RWN spacetimes. 

 Second, for the Hoffmann spacetime of a nucleus of charge $Ze$, one has
\begin{equation}
\cE(r) =  \frac{b^2}{4\pi}\! \int_r^\infty\!\! \Bigl(\sqrt{1 +\tfrac{1}{b^2}\tfrac{Z^2e^2}{s^4} } -1 \Bigr)4\pi s^2 \dd s 
\sim \left\{\!\!
\begin{array}{rl}
\frac12\frac{Z^2e^2}{r}  &\mbox{as}\; r\to\infty,\\
 b^{\frac12}(Ze)^{\frac32}\left( \frac16 \mathrm{B}\left(\frac14,\frac14\right) -
 \left(\frac{b}{Ze}\right)^{\frac12} r\right) &\mbox{as}\;  r\to 0, 
\end{array}
\right.\hspace{-10pt}
\end{equation}
where B$(x,y)$ is Euler's Beta function, and $b>0$ Born's field strength constant.
 In order not to have a black hole, the radial mass function $m(r) = M - \frac{1}{c^2}\cE(r)$ must have a non-positive limit 
when $r\searrow 0$. 
 For assume that $m(0) >0$, then $f^2(r) \searrow -\infty$ as $r\searrow 0$, while $f^2(r)\to 1$ as $r\nearrow\infty$, which
means that there is at least one real $r>0$ at which $f^2(r)$ vanishes. 
 This implies a lower bound on $b$, namely 
\begin{equation}\label{MbareNOTpos}
{\forall\ Z:}\quad 
b \geq \frac{A(Z,N)^2\mPR^2 c^4}{ Z^3e^3\left( \frac16 \mathrm{B}\left(\frac14,\frac14\right)\right)^2}{,}
\end{equation}
{which is} a necessary condition for not to have a black hole in the spacetime {whatever the value of $Z$.
 Since $Z\geq 1$ and $A(Z,N)\leq 3Z$, replacing $A(Z,N)/Z^3$ by $9$ at r.h.s.(\ref{MbareNOTpos}) yields a lower
bound on $b$, uniformly in $Z$.}
 A sufficient condition would guarantee that as long as $m(r)\geq 0$, then $2\frac{G}{c^2}\frac{m(r)}{r} < 1$. 
 When $m(0) =0$ we can state the following sufficient criterion for not having a black hole in the spacetime,
based on the fact that one can show that $m(r)/r > 0$ is bounded and monotonic decreasing, with limit as $r\searrow 0$
given by $bZe$, with $b$ given by r.h.s(\ref{MbareNOTpos}). 
 Thus, if $m(0)=0$ then we have no black hole if $bZe < c^4/2G$, but $b$ is given by r.h.s(\ref{MbareNOTpos})
when $m(0)=0$, and this yields the necessary and sufficient condition
\begin{equation}\label{MbareNOTposZ}\textstyle 
 1 < \frac12 \left( \frac16 \mathrm{B}\left(\frac14,\frac14\right)\right)^2 \frac{Z^2}{A(Z,N)^2}\frac{e^2}{G \mPR^2}
\end{equation}
for the absence of a black hole when $m(0)=0$, {given $Z$}. 
 For all the known nuclei the condition is clearly met, i.e. we are once again deep in the naked singularity sector, 
this time of the Hoffmann spacetimes.
 Lastly, $f^2(r)$ increases when $b$ increases from r.h.s(\ref{MbareNOTpos}) and all other parameters are kept fixed, 
so it follows that we stay in the naked singularity sector of the Hoffmann spacetimes if the central singularity has 
negative bare mass $m(0)<0$. 

 The RWN spacetime and the Hoffmann spacetime with zero bare mass, both in their naked singularity sectors, 
may be seen as the extreme members of {the Hoffmann family of spacetimes with a naked singularity of
negative bare mass, which is included in a larger} family of electrostatic spacetimes with naked singularity and 
negative bare mass discussed in this paper; {see Appendix~\ref{app:electrovac}}.
 In the next section we formulate the Dirac operator for a test electron in 
such spacetimes, then state and prove our theorems about {these Dirac operators}.

\section{The Dirac {Hamiltonian} for hydrogen and hydrogenic ions}\label{sec:Dirac}

\subsection{Test electron without anomalous magnetic moment}

 Due to the spherical symmetry and static character of the spacetimes, the 
Dirac operator $H$ of a test electron in the curved space whose line element $\dd s^2$ is given by (\ref{dsSQR}) 
separates in the spherical coordinates and their default spin frame \cite{CP}. 
 More precisely,  $H$ is a direct sum of so-called partial-wave Dirac operators
$H^{\mbox{\tiny{rad}}}_{k}$ which act on two-dimensional bi-spinor subspaces.
 This reduces the spectral problem to studying the family of radial Dirac operators
$H^{\mbox{\tiny{rad}}}_{k} := \mEL c^2 K^{\mbox{\tiny{}}}_{k}$, $k\in\Zset\backslash\{0\}$, with
\begin{align}\label{Hrad}
K^{\mbox{\tiny{}}}_{k} :=
\left[\begin{array}{cc} f(r) - \frac{e}{\mEL c^2} \phi(r) &
  \frac{\hbar}{\mEL c}\left[ \frac{k}{r} f(r) - f^2(r) \partial_r \right] \\
 \frac{\hbar}{\mEL c}\left[ \frac{k}{r} f(r) + f^2(r) \partial_r \right]
& -  f(r) - \frac{e}{\mEL c^2}  \phi(r) \end{array}\right],
\end{align}
acting on $g(r):=\big(g_1(r) , g_2(r)\big)^T$, {with $g_1$ and $g_2$ $C^\infty$ 
functions compactly supported away from $r=0$},  equipped with a weighted $L^2$ norm given by
\be \label{fnorm}
\| g\|^2 := \int_{0}^{\infty} \frac{1}{f^2(r)}  \Big( |g_1(r)|^2 + |g_2(r)|^2 \Big) \dd r. 
\ee
 Note that $K^{\mbox{\tiny{}}}_{k}$ is physically dimensionless.

 To state and prove our theorems we {make a few assumptions on $f(r)$ and $\phi(r)$ which are satisfied 
by the finite-bare-mass electromagnetic spacetimes in \cite{TZ}, but also by some 
more general spacetimes with negative infinite bare mass, of which the RWN spacetime in its naked singularity sector is but one member}
{(See Appendix~\ref{app:electrovac})}.
 By (\ref{dsSQR}), assumptions on $f(r)$ are equivalent to assumptions on $m(r)$. 
\newpage

\begin{asmps}\label{massASS}
$\phantom{nix}$

\begin{itemize}
\item $m(r)$ is continuous;
\item {${m(r)}/{r} < {c^2}/{2G}$; } 
\item $m(r)\sim -\mNULL r^{- \alpha} $ as  $r \searrow 0$,  where $\alpha \geq 0$ and $ \mNULL >0$;
\item $m(r) \to M >0$ as $ r \rightarrow \infty$.
\end{itemize}
\end{asmps}
 With the above specification of the radial mass function, one has  $f^2(r) \neq 0$ for all $r$, and
\begin{equation} \label{fasmp}
f^2(r)\sim  2\frac{G}{c^2} \frac{\mNULL}{r^{1+ \alpha}}\quad\mbox{as}\quad r \rightarrow 0^+;
 \qquad\qquad f^2(r) \to 1 \quad\mbox{as}\quad  r\rightarrow \infty.
\end{equation}

\begin{rmk}
 Our first two assumptions on $m(r)$ are equivalent to ruling out black holes in spacetimes with the line element (\ref{dsSQR}).
 So our spacetimes feature a charged naked singularity.
 By the third and fourth assumptions on $m(r)$, the naked singularity has a negative
bare mass ($\lim_{r\downarrow 0} m(r)$), which is finite only for $\alpha =0$, 
in which case $m(0) =-C_0$. 
\end{rmk}

 The function $\phi(r)$ is the potential of the electrostatic field generated by the nucleus, in the 
sense that the radial component $E(r) = -\pr\phi(r)$.
 We will make the following assumptions.
\begin{asmps}\label{potASS}
$\phantom{nix}$

\begin{itemize}
\item $\phi(r)$ is continuously differentiable;
\item $\phi(r) \sim {Ze}/{r}$ as $r \rightarrow \infty$; 
\item $\phi(r) \sim {C^{\prime\prime}_{\beta}+}C^{\prime}_{\beta}r^{-\beta}$ around zero, with $\beta\leq 1$.
\end{itemize}
\end{asmps}

\begin{rmk}
 The second assumption on $\phi(r)$ expresses Gauss' law in our asymptotically flat spacetimes.
\end{rmk}
\begin{rmk}
 We note that  $\beta =1$ for the RWN spacetime.
\end{rmk}

 The following is the first main theorem of this section. 

\begin{thm}  \label{th:self} 
Under the stated assumptions on $m(r)$ and $\phi(r)$ the operator $H^{\mbox{\tiny{rad}}}_{k}$ has uncountably many self-adjoint 
extensions, $\forall k\in\Zset\backslash\{0\}$.
\end{thm}

\begin{rmk} Inspection of our proof will reveal that we can generalize our Theorem \ref{th:self} 
and still conclude, with the same proof, that the operator ${H}_k^{rad}$ has multiple self-adjoint extensions if we allow
$\alpha > -1$ and $ \beta < 1+\frac{\alpha}{2}$.
 However, for $\alpha<0$ the mass function $m(r)$ is monotone decreasing in a right neighborhood of $r=0$, which disqualifies
it from the roster of mass functions for the electrostatic spacetimes considered in \cite{TZ}, and their
generalization considered here. 
 Also, { $\beta >1$ does not occur in our electrostatic spacetimes; see Appendix~\ref{app:electrovac}}.
\end{rmk} 
\begin{rmk}{
 Since the third bullet point in Assumption \ref{potASS} covers electric potentials which are bounded at the origin, as
well as those which diverge like an inverse power law when $r\searrow 0$, it is natural to suspect that the conclusion
of our Theorem \ref{th:self} will also hold if we allow $\phi(r)$ to diverge, but less strongly than an inverse power law, 
e.g. logarithmically, when $r\searrow 0$. 
 The proof of the so-modified Theorem \ref{th:self} requires only minor adjustments.}
\end{rmk} 

\begin{proof}[Proof of Theorem~\ref{th:self}]
 Under our assumptions on the mass function, we can change variables $r\mapsto x$ as follows, 
\begin{align} \label{eq:cv}
{\frac{dr}{dx}= \frac{\hbar}{\mEL c}f^2(r(x)) ,}
\end{align}
and study 
\begin{align} \label{Ktilde}
\widetilde{K}_k  &= 
{
\left[\begin{array}{cc} f(r(x))  -\frac{e}{\mEL c^2} \phi(r(x)) &  \frac{\hbar} {\mEL c}\frac{k}{r(x)} f(r(x)) -\frac{d}{dx} \\
\frac{\hbar} {\mEL c}\frac{k}{r(x)} f(r(x))  +    \frac{d}{dx} & - f(r(x)) -\frac{e}{\mEL c^2} \phi(r(x)) 
	\end{array}\right] } \\
	 & =: 
 \left[\begin{array}{cc}  a(x) - b(x)  & {k}c(x) - \frac{d}{dx} \\   {k}c(x) + \frac{d}{dx} &  - a(x) - b(x) 
	\end{array}\right] 
\end{align}
with the inner product 
\begin{align}\label{inner}
 \la g,h   \ra : =   \int_{0}^{\infty} \Big( g_1(r(x)) \bar{h}_1(r(x)) + g_2(r(x)) \bar{h}_2(r(x)) \Big) dx {.}
\end{align}

 Let $C$ denote a generic constant.
 One can show that as $ r \rightarrow 0^+$ $( x \rightarrow 0^{+})$ we have
$r\sim {C}x^{ \frac{1}{2+\alpha} }$, and as $ r \rightarrow \infty $ $( x \rightarrow \infty)$, we have $ r \sim x $.
 Therefore, $a(x) \sim  {C}x^{-\frac{1+\alpha}{4 + 2 \alpha}} $ and
$c(x) \sim  {C}x^{-\frac{3+\alpha}{4+2 \alpha}} $ as $x \rightarrow 0^{+} $,
and $a(x) \sim 1 $ and $c(x) \sim x^{ - 1} $ as $x \rightarrow \infty $. 
 One further has  $b(x)\sim  {C}x^{-\frac{\beta}{2+ \alpha}} $ as  $x \rightarrow 0^{+} $, and 
$b(x) \sim  x^{-1} $ as  $x \rightarrow \infty $. 

 Let $\widetilde{K}^*_k$ be the adjoint operator of $\widetilde{K}_k$. 
 Recall that $\cD(\widetilde{K}_k)$ is all $C^{\infty}$ functions of compact support in $(0,\infty)$,
 and $\cD(\widetilde{K}^{0*}_k)$ includes the functions $f$ so that $f$ and $f^{\prime}$ are integrable in any 
compact subset of $[0,\infty)$, see Section~XIII in \cite{DS}. 
 We define the following  sesquilinear form on $\cD(\widetilde{K}^*_k)$: 
\be
 [g, h]: = \la \widetilde{K}_k ^* g, h \ra - \la g, \widetilde{K}_k ^*h \ra ,
\ee
where $ \la \cdot, \cdot  \ra$ is defined as in \eqref{inner}. 
 By Theorem~4.1 in \cite{WeiC}, $ \widetilde{K}^*_k |_{\cD}$ is a self-adjoint extension of $\widetilde{K}_k$ iff 
 \begin{enumerate}
 \item[$i)$]$ \cD(\widetilde{K}_k) \subset  \cD \subset \cD( \widetilde{K}^*_k)$
\item[$ii)$]  $[g, h] =0$ for all $g,h \in \cD$ 
 \item[$iii)$] if $g \in \cD( \widetilde{K}^*_k)$ and  $[g, h] =0$ holds for every $h \in \cD$ then $g \in \cD$. 
 \end{enumerate}

 We first start considering the spaces in which $ [g, h]=0$. 
 Take $ g \in \cD(\widetilde{K}_k ^*)$, i.e. $ \widetilde{K}_k ^* g  = \psi $ for some $\psi \in L^2([0, \infty)$. 
 Since $ \cD(\widetilde{K}_k) \subset C_c^{1}([0, \infty))$, $g \in AC([x_1,x_2])$ for each $ 0 \leq x_1 <x_2< \infty$, and therefore we can 
integrate to get
\begin{align}\label{ga}
g_1(x) = e^{-\mu(x)} \Big( g_1(0) + \int_0^{x} e^{+\mu(y)} [(  a(y) +  b(y) )g_2(y) + \psi_2(y)] dy \Big),  \\
g_2(x) = e^{+\mu(x)} \Big( g_2(0) + \int_0^{x} e^{-\mu(y)} [(  a(y) -  b(y) )g_1(y) - \psi_1(y)] dy \Big)  
\label{gb}
\end{align}
for each $ x \leq x_2 < \infty$, where $\mu(x) = \int_0^{x} {k}c(y) dy \sim C x^{\frac{1+\alpha}{4+ 2 \alpha}} \rightarrow 0$ 
as  $x \rightarrow 0$.
 To have $g_1$, $g_2$ to be defined we also need   $a(x), b(x) \in L^2 ([0, x_2])$, $x_2 < \infty$.
\begin{rmk}\label{alphabeta}
 For our electrostatic spacetimes, $\alpha\geq 0$ and $\beta\leq 1$, so we do have $a(x), b(x) \in L^2 ([0, x_2])$, $x_2 < \infty$.
 If we drop the `electrostatic' requirement and also allow negative $\alpha >-1$, then this requires $ \beta < 1 + \frac{\alpha}{2}$. 
 Cf. our earlier remark.
\end{rmk}

By integration by parts, we have
\begin{align}
[g, h] 
& = \lim_{x_1 \rightarrow 0} \lim_{x_2 \rightarrow \infty}  \int_{x_1}^{x_2} {\Big(} ( \widetilde{K}_k ^* g)_1 \overline{h_2} - 
(\widetilde{K}_k ^* g)_2 \overline{h_1} + g_1 \overline{( \widetilde{K}_k ^*h)_2} - g_2 \overline{( \widetilde{K}_k ^*h)_1} \ra \Big) 
{dx}
\\
& = \lim_{x_1 \rightarrow 0} \lim_{x_2 \rightarrow \infty} \Big[ g_1(x_2) \overline{h_2(x_2)} - g_2(x_2) \overline{h_1(x_2)} -  
g_1(x_1) \overline{h_2(x_1)} + g_2(x_1)\overline{h_1(x_1)} \Big] \\   & = g_2(0) \overline{h_1(0)} - g_1(0) \overline{h_2(0)}.
\end{align}
 We used \eqref{ga}, \eqref{gb}, and the fact that $g,h \in L^2([0,\infty))$  to obtain the last equality. 

 This suggests that any symmetric extension requires $g_2(0) \overline{h_1(0)} - g_1(0) \overline{h_2(0)}= 0$.
 Taking $g=h$, one can see this is true if $g_1(0)$ is a real multiple of $g_2(0)$, or vice versa.
 Therefore, for any $ 0 \leq \theta < \pi$, 
\begin{align} \label{Domdef}
{\widetilde{K}}_{k;\theta} =  
\widetilde{K}_k ^*| \cD_\theta,\,\, \text{where}  \,\ 
\cD_\theta = \{g\in \cD(\widetilde{K}_k ^*): g_1(0) \sin \theta + g_2(0) \cos \theta =0\} 
\end{align}
gives a symmetric extension, cf. \cite{CP}. 
 Note that $\cD_\theta$ satisfies both the conditions $i)$ and $ii)$. Condition $iii)$ is also clearly satisfied. Let $h \in \cD_\theta$
  \begin{align}
 [g,h] =0 \iff g_2(0) \overline{h_1(0)} - g_1(0) \overline{h_2(0)} =0 \iff
 \frac{\overline{h_2(0)}}{\overline{h_1(0)} } = \frac{g_2(0)}{ g_1(0)} = - \tan \theta \in \R.
  \end{align} 

 This finishes the proof.  \end{proof}

\begin{rmk}
It is worth mentioning that,  since $[g,h]$  is the difference of two positive rank-one bilinear forms, 
$\widetilde{K}_{k}$, and thus $K_{k}$, has deficiency indices $(1,1)$.
\end{rmk}

 Recall that in the partial wave decomposition the Dirac Hamiltonian $H$ is a direct sum of operators 
$H_{k}^{\text{rad}}= {\mEL c^2} K_{k}^{\text{}}$ which act on two-dimensional bi-spinor subspaces.
 Having shown that these have self-adjoint extensions $K ^{\mbox{\tiny{}}}_{k,\theta}$ for $ 0 \leq \theta < \pi$,
we now define $ H_\theta$, as the direct sum of the 
$H_{k;\theta}^{\text{rad}}:={\mEL c^2} K^{\mbox{\tiny{}}}_{k,\theta}$.
 We are ready to state our next main theorem of this section.
\begin{thm}\label{th:esspec} Under Assumptions~\ref{massASS} and \ref{potASS}, for hydrogenic ions we have

\begin{itemize}
\item[(a)]
  The essential spectrum $\sigma_{ess}(H_{\theta}) = (-\infty,-\mEL c^2]\cup [\mEL c^2, \infty)$;
\item[(b)] $ H_{\theta}$ has purely absolutely continuous spectrum in $(-\infty,-\mEL c^2 ) \cup (\mEL c^2, \infty)$;
\item[(c)] the singular continuous spectrum $\sigma_{sc} (H_{\theta})=  \emptyset$.
\end{itemize}
\end{thm}
\begin{proof}
 We prove Theorem~\ref{th:esspec} by a series of Lemmas, Corollaries, and Remarks below.
 In particular, Lemma~\ref{essK} establishes part $(a)$;
this together with Lemma~\ref{acK} and Corollary~\ref{cor:sc} establishes parts $(b)$ and $(c)$. 
\end{proof}

\begin{lemm} \label{essK}
$\sigma_{ess} ( K_{k;\theta}^{\text{}}) = {( - \infty , - 1] \cup [ 1, \infty)}$.
\end{lemm} 
 To prove Lemma \ref{essK} we recall the following lemma from \cite{CP}. 
\begin{lemm}\label{cplem} Let 
\be
{D :}= \left[ \begin{array}{cc} 0 & - \frac{d}{dx} \\ \frac{d}{dx} & 0  \end{array} \right] 
\ee
{be defined on the $C^{\infty}$ two-component functions of compact support in the positive real half-line.
 Now take the closure of this operator in $(L^2(\Rset_+))^2$ with the boundary condition $f_1(0) \sin \theta + f_2(0) \cos \theta=0$ 
at $x=0$, {denoted} $D_\theta$.}
 Let $A$ be the operator 
\be
A = \left[ \begin{array}{cc}a_{11}(x) & a_{12}(x) \\  a_{21}(x) &  a_{22}(x)	  \end{array} \right] ,
\ee
where the $a_{ij}$ are functions in $L^2([0,x_2])$ for all $ 0 <x_2 < \infty$ and $a_{ij}(x) \rightarrow 0$ as $ x \rightarrow \infty$. 
 Then $A$ is $D_{\theta}$ compact.  
\end{lemm}
\begin{rmk}
  The value $x=0$ in this lemma plays no role in its proof. 
  In particular, one can pick $x=x_1$, $0 <x_1 < x_2 < \infty$, and consider the operator $D$ with boundary condition 
 $f_1(x_1) \sin \theta + f_2(x_1) \cos \theta=0$ in $L^2([x_1,x_2])^2$.
\end{rmk}
\begin{proof}[Proof of Lemma~\ref{essK}] 
 We split the operator $\widetilde{K}_k$ in \eqref{Ktilde} as
\begin{align}
\widetilde{K}_k  = \left[ \begin{array}{cc} 1 & kc(x) - \frac{d}{dx}  \\ 
kc(x) + \frac{d}{dx}  & - 1
	  \end{array} \right] + \left[\begin{array}{cc} \Big[a(x) - 1 \Big] -b(x) &   0 \\
  0 & -  \Big[a(x) - 1\Big] -b(x) 
	\end{array}\right]  
=:  \widetilde{K}^0_k + V.
\end{align}
 Note that Theorem~\ref{th:self} is valid when $a(x)=1$ and $b(x)=0$.
 Therefore, $\widetilde{K}^0_k$ has deficiency indices $(1,1)$ and has multiple self-adjoint extensions similar to $\widetilde{K}_k$.
 We define these self-adjoint extensions as $\widetilde{K}^0_{k;\theta}$ similar to $\widetilde{K}_{k;\theta}$, see \eqref{Domdef}. 
 We define the following  Weyl sequence for $\widetilde{K}^0_{k;\theta}$, with any $\lambda\in\Rset$,
$|\lambda| > {1}$: 
\be \label{weyl}
f_{n, \lambda} (x) = \frac{1}{2 n^{\f 32}} x e^{- \frac{x}{2n} + i x \sqrt{ \lambda^2 - 1}} 
\left[\begin{array}{c} \sqrt{1 + \frac{1}{ \lambda } } \\ i \sqrt{1 - \frac{1}{ \lambda } }
	\end{array}\right];\quad n\in\Nset.
\ee
 We have that $ \| f_{n, \lambda} \|_{(L^2(\Rset_+))^2} =1$, $  f_{n, \lambda} (x) \rightarrow 0 $ weakly, and 
$\| ( \widetilde{K}^0_{k;\theta} - \lambda) f_{n, \lambda} (x) \|_{(L^2(\Rset_+))^2} \rightarrow 0 $ as $n\to\infty$.
 Hence, any $ |\lambda| > {1}$ 
is in the essential spectrum of $\widetilde{K}^0_{k;\theta}$.
 Further, since the essential spectrum is a closed subset of $\R$, one has 
$ ( - \infty , - 1] \cup [1, \infty)\subset \sigma_{ess} ( \widetilde{K}^0_{k;\theta})$.

 For the reverse inclusion, we consider the operator  $[\widetilde{K}^0_{k;\theta}]^2$. 
 Let $g \in \cD_\theta$, then one has 
\begin{align}\label{reverseA}
 \la [\widetilde{K}^0_{k;\theta}]^2 g, g \ra & =  
\la \widetilde{K}^0_{k;\theta} g, \widetilde{K}^0_{k;\theta} g \ra \\ 
\label{reverseB}
& =
 \int_0^{\infty} \left|\left( \frac{d}{dx} + {k}c(x)\right) g_1\right|^2  dx 
+ \int_0^{\infty} \left|\left(- \frac{d}{dx} + {k}c(x)\right) g_2\right|^2 dx \\ \notag
& \qquad\qquad\qquad\qquad + \|g \|^{2}_{(L^2(\Rset_+))^2} +  {\sin (2\theta)\left(|g_1(0)|^2 + |g_2(0)|^2 \right);} 
\end{align}
 For the boundary term at the r.h.s. of this equality, recall the boundary conditions in $\cD_\theta$.
 Clearly if $ 0 \leq \theta \leq \frac{\pi}{2}$, then $ \la [\widetilde{K}_{k;\theta}]^2 g,g \ra \geq  \la g,g \ra$. 
Therefore, if $ 0 \leq \theta \leq \frac{\pi}{2}$ then 
$ \sigma( \widetilde{K}^0_{k;\theta})= \sigma_{ess} ( \widetilde{K}^0_{k;\theta}) = ( - \infty , - 1] \cup [1, \infty)$. 
On the other hand, all self-adjoint extensions of $\widetilde{K}^0_{k}$ have the same essential spectrum, cf. \cite{WeiC}, p.163.
 Therefore, $\sigma_{ess} ( \widetilde{K}^0_{k;\theta}) = ( - \infty , - 1] \cup [1, \infty)$  for all $\theta\in[0,\pi)$.

 Next we will show that $V$ is $\widetilde{K}^0_{k;\theta}$  compact. 
 We define $ \xi(x) = \int_0^x {k}c(y) dy$ for $0 \leq x \leq 1$ and $ \xi(x) = \xi(1)$ for $x>1$.
 Then the matrix 
\be\label{S}
S= \left[ \begin{array}{cc} e^{-\xi(x)} &0 \\  0 &  e^{\xi(x)}
	  \end{array} \right]  
\ee
is bounded.   

Assume that $\|g_n\|_{(L^2(\Rset_+))^2} , \| \widetilde{K}^0_{k;\theta} g_n\|_{(L^2(\Rset_+))^2} $ are bounded sequences.
 Then $\| Sg_n \|_{(L^2(\Rset_+))^2}$ and  $\| D_\theta S g_n\|_{(L^2(\Rset_+))^2}$ {are also bounded,
the first one is because $S$ is bounded and the latter one is by the fact that 
$ D_\theta S = S^{-1} S D_\theta S =  S^{-1} (\widetilde{K}^0_{k;\theta} +W)$ for some bounded $W$.}
 Moreover, one can check that $VS^{-1}$ is $D_\theta$ compact by Lemma~\ref{cplem}. Hence, 
\be 
 VS^{-1} S g_n = Vg_n
\ee 
has a convergent subsequence. 
 This proves that $V$ is $D_{\theta}$ compact. 
\end{proof}

\begin{lemm}\label{acK} $\sigma(K_{k;\theta}^{\text{}})$ is purely absolutely continuous in 
$( - \infty , - 1)\cup( 1, \infty)$.
\end{lemm}

For the proof we utilize the following Theorem from \cite{WeiB,WeiC}.
\begin{thm}\label{th:weid} (Weidmann) Let 
\be
\tau:=  \left[\begin{array}{cc}  0  & - \frac{d}{dx} \\   \frac{d}{dx} &  0
	\end{array}\right] + P_1(x) + P_2(x) 
\ee
be defined on $ (a , \infty)$. 
 Further assume that  $|P_1(x)| \in L^1 (c, \infty)$  for some $ c \in (a , \infty)$, and  $P_2(x)$ is of 
bounded variation in $[c, \infty)$ with 
\be 
\lim_{x \rightarrow \infty} P_2(x) = \left[\begin{array}{cc}  \mu_{+}  & 0 \\   0 &  \mu_{-}
	\end{array}\right] \,\,\,\,\,\ \text{for} \,\,\,\,\,\,  \mu_{-} \leq  \mu_{+} .
\ee
Then every self-adjoint realization $A$ of $\tau$ has purely absolutely continuous spectrum in 
$(-\infty, \mu_{-}) \cup (\mu_{+}, \infty)$. 
\end{thm}

\begin{proof}[Proof of {Lemma~\ref{acK}}] 
Recall that {$\widetilde{K}^{}_{k,\theta}$} is in the form of $\tau$ with $P_1(0) =0$ and 
\begin{align} \label{def:p2}
{P_2(x) =  \left[\begin{array}{cc} f(r(x)) -\frac{e}{\mEL c^2} \phi(r(x)) & \frac{\hbar}{\mEL c} \frac{k}{r(x)} f(r(x)) \\
\frac{\hbar}{\mEL c} \frac{k}{r(x)} f(r(x))  & - f(r(x)) -\frac{e}{\mEL c^2} \phi(r(x)) 
	\end{array}\right].}
\end{align}
 By hypothesis, both $\phi(r(x))$ and $m(r(x))$ are continuously differentiable and hence of bounded variation. 
Furthermore, we can see that
\be
\lim_{x \rightarrow \infty} \phi(r(x))=0 = \lim_{x \rightarrow \infty} \frac{k}{r(x)} f(r(x){)}, 
\,\,\,\,\mbox{and} \,\,\,\, 
  \lim_{x \rightarrow \infty} f(r(x){)} = 1,
\ee 
or, equivalently,
{
\be \label{lim0}
\lim_{x \rightarrow \infty} P_2(x) = \left[\begin{array}{cc}  1 & 0 \\   0 &  - 1	\end{array}\right]. 
\ee
 Hence the spectrum of $\widetilde{K}_{k;\theta} $ is purely absolutely continuous in 
$( - \infty , - 1) \cup (1, \infty)$.}
\end{proof}

 The two lemmas imply the following.
\begin{cor}\label{cor:sc} 
  The singular continuous spectrum $\sigma_{sc}(K_{k;\theta}^{\text{}})= \emptyset$.
\end{cor}
\begin{proof}
 By Lemma~\ref{acK} and Lemma~\ref{essK} the essential spectrum is the closure of $\sigma_{ac}(K_{k,\theta})$. 
 Since the singular continuous spectrum is a subset of the essential spectrum, and since the interior of the
essential spectrum here is purely absolutely continuous, a non-empty $\sigma_{sc}(K_{k;\theta}^{\text{}})$ 
would have to consist of the discrete set $\{-1, 1\}$, which is impossible.
\end{proof}

Next we turn to the discrete spectrum.  

\begin{thm} \label{thm:eigenvalues} Under Assumptions~\ref{massASS} and \ref{potASS}, 
the eigenvalues of any self-adjoint extension $H_\theta$ of $H$ form a countably infinite set 
located in the gap of the essential spectrum. 
 The set of accumulation points of this discrete spectrum $\sigma_{disc}(H_\theta^{})$ is either $\{\mEL c^2\}$ or
$\{-\mEL c^2,\mEL c^2\}$; for the empirically known hydrogenic ion parameters, only $\mEL c^2$ is an 
accumulation point of the discrete spectrum. 
\end{thm}

 To validate this theorem we utilize Theorem~2.3 from \cite{HMRS}. 
\begin{thm}\label{thm:Hinton} (Hinton et al.) 
 Let 
\begin{align}
Ly:= \left[\begin{array}{cc} 0 & -1  \\ 1 & 0 	\end{array}\right] \left\{ y^{\prime} -
 \left[\begin{array}{cc} p(x)  & c_1 + V_1(x)  \\c_2 - V_2(x) & -p(x) 
	\end{array}\right]y    \right\}  =: 
J y^{\prime}  -  P y.
\end{align}
 Assume that $d > 0$.
 Let $g$ be a nontrivial positive linear functional, and assume $P$ is locally absolutely continuous.
 Let $L_1$ be any self-adjoint extension of $L$. 
 Then $ \sigma (L_1) \cap ( -d,d) $ is infinite if the scalar differential equation,
\begin{align} \label{eq:secord}
-g(I)z^{\prime \prime} + g \Big( P^2 - d^2 I + \frac{ [P^{\prime}J - J P^{\prime} ]}{2} \Big) z =0 
\end{align}
is oscillatory either at $0$ or at $ \infty$. 
\end{thm} 

 In \eqref{eq:secord}, $g$ is a nontrivial linear positive  functional defined on the real $n \times n$ matrices. Therefore, 
for any symmetric and positive semidefinite operator $B$, one has $ g(B) = \sum_{i=1}^n {\delta_i}\la Bu_i , u_i \ra $,
where {the $\delta_i\geq 0$ sum to 1, and the} $u_i$ are non-zero {orthonormal} $n$-vectors. 
In our case, {$n=2$, with $u_1=( 1,0)^T$ and $u_2=(0,1)^T$, and we will consider the ``$+$ case'' 
with $\delta_1^+=1$ and $\delta_2^+=0$, and the ``$-$ case'' with $\delta_1^-=0$ and $\delta_2^-=1$.}

\begin{proof}[Proof of Theorem~\ref{thm:eigenvalues}] 
 We use Theorem \ref{thm:Hinton} for $P = - P_2$, where $P_2$ is as in \eqref{def:p2} and $c_1= c_2 = d=  1$. 
 In particular, we have
\begin{align}\label{Vp}
&V_1= \frac{e}{\mEL c^2} \phi(r(x)) +  f(r(x)) - 1, \\ 
& V_2= \frac{e}{\mEL c^2} \phi(r(x)) -  f(r(x)) + 1,  \\ 
& p(x) = - \frac{\hbar}{\mEL c} \frac{k f(r(x)) }{r(x)}. 
\end{align}
Note that $\frac{r}{x} \rightarrow \frac{\hbar}{\mEL c}$ as $ x \rightarrow \infty$. 
Hence $f(r(x)) \sim 1 - \frac{GM\mEL}{\hbar c x}$ as $ x \rightarrow \infty$.

 Therefore, in the ``$+$ case'', resp. the ``$-$ case'', equation \eqref{eq:secord} yields    
\begin{align}\label{zprpr}
- z^{\prime \prime}+\Gamma_{\pm}(x)z =0 ,
\end{align}
where 
\begin{align}
\Gamma_+(x) = V_2^2(x) - 2 V_2(x) + p^2 + p^{\prime}  \label{GammaA}, \\
\Gamma_-(x) = V_1^2(x) + 2 V_1(x) + p^2 - p^{\prime} \label{GammaB}.
\end{align}
 This gives 
\begin{equation}\label{gammaasymp}
 \Gamma_{\pm} \sim - 2
 \Big[ \tfrac{GM \mEL}{\hbar c} \pm \tfrac{Ze^2}{\hbar c}\Big]\tfrac1x . 
\end{equation}
 Clearly, $ \lim_{x \rightarrow \infty} x^2  \Gamma_+(x) < - \frac{1}{4}$.
 Hence, equation (\ref{zprpr}) with ``$+$'' has solutions with oscillatory behavior at infinity, see \cite[Section~XIII]{DS}. 

 Thus by Theorem \ref{thm:Hinton}, $\sigma(\widetilde{K}_{k,\theta})\cap (- 1,1)$ is infinite,  
so $\widetilde{K}_{k,\theta}$, and therefore ${K}^{\text{}}_{k,\theta}$, have infinitely many eigenvalues in the gap of their
essential spectrum. 

 Lastly we recall that the spectrum of $H_\theta^{}$ is the union of the spectra of the radial Dirac operators obtained by the
partial wave decomposition. 
 This proves that the discrete spectrum of $H_\theta^{}$ is infinite and located in the gap $(-\mEL c^2,\mEL c^2)$.

 We next prove our statement about its accumulation points. 
 We will need the following.

\begin{lemm}\label{lem:accum} 
For each operator $K_{k,\theta}$ the set of accumulation points of its discrete spectrum is
$\{1\}$ or $\{-1,1\}$, depending on whether ${GM \mEL}<{Z e^2}$ or ${GM \mEL}>{Z e^2}$, respectively.
\end{lemm}
\begin{proof}
To determine if $ \pm 1$ are cluster point of the eigenvalues, we consider the operator $L$  in 
Theorem~\ref{thm:Hinton} for $ c_1 = 1 \pm \epsilon$, $ c_2 = 1 \mp \epsilon$ for some $0 < \epsilon <1$. 
In particular,  $ \pm 1 $ is not an accumulation point if and only if  \eqref{eq:secord} for $d = 1 \pm \epsilon$, 
$c_1$ and $c_2$  is non-oscillatory at both $0$ and $ \infty$, see \cite[Theorem~4.1]{HMRS}. 
 
 We first show that $1$ is an accumulation point. Plugging $ c_1 = 1 + \epsilon$, and $d=c_2 = 1 - \epsilon$ together 
with the functions in \eqref{Vp} into \eqref{eq:secord} we obtain \eqref{zprpr} with $\Gamma_\pm^{}$ replaced by $\Gamma_\pm^1$, where
 \begin{align} \label{eq:acpoint}
 & \Gamma^1_+  :=V_2^2 - 2(1- \epsilon) V_2 + p^2 + p^\prime ,\\ 
 & \Gamma^1_-  := V_1^2 + 2(1 +\epsilon) V_1 +p^2 - p^\prime + 4 \epsilon. 
 \end{align}
 It suffices to show that equation \eqref{zprpr} with $\Gamma^1_+$ has an oscillatory solution either at $\infty$ or $0$. 
 An easy calculation shows that, by \eqref{gammaasymp}, one has
 \be 
\Gamma^1_+ \sim - 2(1-\epsilon) \Big[ \tfrac{GM \mEL}{\hbar c} +\tfrac{Z e^2}{\hbar c}\Big]\tfrac1x \,\,\,\,\,\,\,\,\  \text{as} \,\, x\rightarrow \infty.
 \ee
 Since $0 < \epsilon <1$, one has $ \lim_{x \rightarrow \infty} x^2 \Gamma^1_+ < - \frac{1}{4}$. 
Therefore, \eqref{zprpr} with $\Gamma^1_+$ has an oscillatory solution near infinity,
and thus $1$ is always an accumulation point of the set of eigenvalues.
  
 Next we show that $-1$ is a cluster point if ${GM\mEL}\!>\!{Z e^2}$ but not if ${GM\mEL}\!<\!{Z e^2}$. 
 Consider \eqref{eq:secord} with $ d=c_1 = 1 -\epsilon$, and $c_2 = 1 + \epsilon$. 
 One obtains \eqref{zprpr} with $\Gamma_\pm^{}$ replaced by $\Gamma_\pm^{-1}$, where
 \begin{align} 
 & \Gamma^{-1}_+:=  V_2^2 - 2(1+ \epsilon) V_2 + p^2 + p^\prime + 4 \epsilon, \label{gamma+-1} \\ 
 & \Gamma^{-1}_-:=   V_1^2 + 2(1- \epsilon) V_2 + p^2 - p^\prime  .\label{gamma--1}
 \end{align} 
 We now find
\be 
 \Gamma^{-1}_-\sim -2(1-\epsilon)\Big[\tfrac{GM\mEL}{\hbar c}-\tfrac{Z e^2}{\hbar c}\Big]\tfrac1x \,\,\,\,\,\,\,\,\ \text{as} \,\, x \rightarrow  \infty.
\ee
 Thus, since $0 < \epsilon <1$, one has 
$\lim_{x \rightarrow \infty} x^2 \Gamma^{-1}_{-}  < - \frac{1}{4} $ if  ${GM\mEL}\!>\!{Z e^2}$, 
and then \eqref{zprpr} with $\Gamma^{-1}_-$ has an oscillatory solution near infinity, so $-1$ is a cluster point of the discrete spectrum.

 On the other hand,  since $0 < \epsilon <1$, one has  $\lim_{x \rightarrow \infty} x^2 \Gamma^{-1}_{-} >0\; (> -\frac{1}{4})$ if 
 ${GM \mEL}\!<\!{Z e^2}$.
 Therefore, the solution to \eqref{zprpr} with  $\Gamma^{-1}_{-}$ is non-oscillatory at $\infty$ if ${GM \mEL}\!<\!{Z e^2}$.
 To see that it is also non-oscillatory at $0$, we consider
\be
p \pm \nu:= p \pm \frac{1}{2} ( V_1 + V_2 +c_1 - c_2) = - \frac{\hbar k}{m_e c} \frac{f(r(x))}{r(x)} \pm \frac{e}{m_e c^2} \phi(r(x)) - 2 \epsilon .
\ee
 By Corollary~3.4 in \cite{HMRS},  \eqref{eq:secord} is non-oscillatory at $0$ if $ p \pm \nu \leq -1$ in one neighborhood of $0$. 
Recall that $\frac{f(r(x))}{r(x)} \sim x^{- \frac{3+ \alpha}{4 + 2 \alpha}}$, and $\phi(r(x)) \sim x^{-\frac{\beta}{ 2 + \alpha}}$ around zero,
with $\alpha\geq 0$ and $\beta \leq 1$. (Or if we allow also $\alpha <0$, then $\beta < 1 + \frac{\alpha}{2}$; cf. Remark~\ref{alphabeta}.) 
 Hence, $ p \pm \nu \rightarrow - \infty$, and so $-1$ is not an accumulation point of the discrete spectrum if ${GM \mEL}<{Z e^2}$.
 \end{proof} 

 Lemma~\ref{lem:accum} and the fact that the spectrum of $H_\theta^{}$ is the union of the spectra of the radial partial wave
Dirac operators now concludes our proof about the accumulation points of the discrete spectrum, for general $M$. 
 For known nuclei $M= M_{\mbox{\tiny{ADM}}}=A(Z,N)\mPR$, 
with $A\leq 3Z$, so $\frac{GM \mEL}{Ze^2} < 2\times 10^{-39}$, and so $-\mEL c^2$ is not an accumulation point of the
discrete spectrum of $H_\theta$ for any empirical hydrogenic ions. 
 The proof of our theorem is complete.
\end{proof}

\begin{rmk}\label{meL}
{We suspect that the boundary points of the essential spectrum, $-\mEL c^2$ and $\mEL c^2$, are generally not eigenvalues of 
 $H_\theta$, but we have not tried to prove it, and the answer may depend on $\theta$ and on the value of ${GM \mEL}/{Ze^2}$.}
\end{rmk}

\begin{rmk} 
 We proved with the actual empirical values of $\frac{GM \mEL}{Z e^2}$ for physical hydrogenic ions that
$\mEL c^2$ is a limit point of the discrete spectrum while $-\mEL c^2$ is not. 
 The appearance of $-\mEL c^2$ as a limit point of the discrete spectrum for {hypothetical hyper-heavy ion} 
values $\frac{GM \mEL}{Z e^2} >1$
can be explained in physics lingo if we recall that the negative continuum is usually interpreted as being associated with
positrons, which do not bind electrically to the positively charged nuclei, but which can be bound gravitationally
if the gravitational attraction to the nucleus overcomes the electrical repulsion. 
 Incidentally, the same  critical value $\frac{GM \mEL}{Z e^2} = 1$ features also in the 
non-relativistic {treatment}, where the Newtonian gravitational attraction between a positron and a 
nucleus overpowers their Coulomb repulsion if and only if $\frac{GM \mEL}{Z e^2} >1$, in which case the
Schr\"odinger Hamiltonian also has infinitely many bound states, while there are no bound states when $\frac{GM \mEL}{Z e^2} \leq 1$. 
 Our results do not reveal whether the general-relativistic Dirac problem in the critical case $\frac{GM \mEL}{Z e^2} = 1$ features any 
bound positron states; our results only show that there are none if $\frac{GM \mEL}{Z e^2} < 1$, and infinitely many if
$\frac{GM \mEL}{Z e^2} > 1$.
 \end{rmk}
 
\subsection{Test electron with anomalous magnetic moment}

In the special-relativistic problem of  hydrogenic ions at any $Z\in\Nset$ 
it was found long ago \cite{Beh,GST} that the addition of an anomalous magnetic moment operator to the Dirac Hamiltonian 
of a test electron in the Coulomb field of the point nucleus suffices to produce an essentially self-adjoint Dirac Hamiltonian.
 For a test electron in the RWN spacetime of a point nucleus it was found in \cite{BMB} that
a \emph{sufficiently large} anomalous magnetic moment of the electron is required to obtain an essentially 
self-adjoint Hamiltonian of the hydrogenic ions; it turns out that 
the empirical electron value is large enough \emph{uniformly} for all $Z\in\Nset$.

 This suggests that adding an anomalous magnetic moment operator to the Dirac Hamiltonian of a test electron {may}
restore essential self-adjointness {also} in all situations discussed here so far where essential self-adjointness fails, 
in particular for the Dirac Hamiltonian of a test electron in the Hoffmann spacetime of a point nucleus with negative bare mass.
 Interestingly, the situation is more complicated, as shown by our next theorem.

The radial partial-wave Dirac operator $H^{\text{rad}}_{\mu_a,k}= \mEL c^2 K^{\text{}}_{\mu_a,k}$ now is given by
\begin{align} 
K_{\mu_a,k}^{\text{}}
:= 
\left[\begin{array}{cc} f(r) - \frac{e}{\mEL c^2}\phi(r)  & 
\frac{\hbar} {\mEL c}\left[\frac{k}{r} f(r)-f^2(r)\partial_r \right] -\frac{\mu_a}{\mEL c^2}  \phi^{\prime}(r)f(r)\\
\frac{\hbar} {\mEL c}\left[\frac{k}{r} f(r) + f^2(r) \partial_r\right] 
-\frac{\mu_a}{\mEL c^2} \phi^{\prime}(r)f(r) &  - f(r) - \frac{e}{\mEL c^2} \phi(r)
	\end{array}\right].
\end{align}

\begin{thm} Let $f^2(r)$ be as in \eqref{fasmp} with $ \alpha \geq 0$ and 
$ \phi(r)= {C_\beta^{\prime\prime}+}\phiNULL r^{-\beta} + O_1( r^{ \frac{3}{2} - \beta} )$ for $\beta \leq 1$,
where $f \in O_k(g) $ indicates $ \frac{d^j}{dr^j} f = O( \frac{d^j}{dr^j} g) $ for $j=0,1,..k$. 
 Then the operator $H^{\text{rad}}_{\mu_a,k}$  is essentially self-adjoint  if either $\beta >  \frac{1+\alpha}{2} $, 
or $\beta =  \frac{1+\alpha}{2}$ and $| \mu_a | \geq  \frac{ 2 + \alpha}{1+  \alpha} \hbar \sqrt{2G\mNULL}/|\phiNULL| $. 
 On the other hand, if $\beta <  \frac{1+\alpha}{2} $, then $H_{\mu_a,k}^{\text{rad}}$ 
has multiple self-adjoint extension.
\end{thm}

\begin{proof} 
 Note that the fact that $H^{\text{rad}}_{\mu_a,k}$ has multiple self-adjoint extension if $\beta <  \frac{1+\alpha}{2} $ is
 a consequence of Theorem~\ref{th:self}.
 In particular, if the change of variable {as in}
 \eqref{eq:cv} is applied to $K^{\text{}}_{\mu_a,k}$, one obtains the operator 
\begin{align}\label{Kmu-x}
\widetilde{K}_{\mu_a,k}  = \widetilde{K}_{k} +
 \mu_a\left[\begin{array}{cc}  0 &   d(x) \\   d(x) & 0
	\end{array}\right],
\end{align}
where $ d(x) \sim {C}x^{-\frac{3 + \alpha + 2 \beta}{4 + 2 \alpha}}$ as $ x \rightarrow 0$, and 
 $d(x) \sim x^{-2}$ as $ x \rightarrow \infty$.
 Therefore, $g_1$ and $g_2$ in \eqref{ga}, \eqref{gb}
 arises with $\mu(x) = \int_0^{x}[kc(y) +\mu_a d(y)] dy$.
  Since,  $ \mu(x) \rightarrow 0 $ if $\beta <  \frac{1+\alpha}{2} $, the proof follows similar to the proof of Theorem~\ref{th:self}.
 Therefore, it remains to prove the assertions for $\beta \geq  \frac{1+\alpha}{2} $.

We start the proof with the case that $\beta >  \frac{1+\alpha}{2} $. 
We will show that the limit point case (LPC) is verified in the right neighborhood of $r = 0$ if $\beta >  \frac{1+\alpha}{2} $, 
 i.e. there is at least one non-square integrable solution to $K^{\text{}}_{\mu_a,k} g = \lambda g$ 
for each $\lambda \in \mathbb{C}$, or equivalently for a fixed $\lambda$, see \cite[Theorem~5.6]{WeiC}.
 In particular, we will consider the solutions to 
\begin{align} \label{0eigen}
&\Big[-\frac{e\phi(r)}{\mEL c^2 f^2(r)}+\frac{1}{f(r)} \Big] g_1 = 
\Big[\frac{\hbar }{\mEL c}\left[ \partial_r - \frac{k}{r  f(r)} \right] + \frac{\mu_a \phi^{\prime}(r)}{\mEL c^2 f(r)} \Big] g_2 , \\
&\quad\; \Big[ \frac{e\phi(r)}{\mEL c^2 f^2(r)}  + \frac{1}{f(r)}  \Big] g_2 = 
\Big[ \frac{\hbar }{\mEL c}\left[\partial_r + \frac{k}{r  f(r)}\right] -   \frac{\mu_a \phi^{\prime}(r)}{\mEL c^2 f(r)} \Big] g_1.
\label{00eigen}
\end{align} 
Recall that $g = (g_1, g_2)^T$ is square integrable in the right neighborhood of $r = 0$ with the inner product \eqref{fnorm} iff 
for each $ 0 <R < \infty$,
\begin{align} \label{sqint}
\int_{0}^{R} \frac{1}{f^2(r)}  \Big( |g_1(r)|^2 + |g_2(r)|^2 \Big) dr < \infty .
\end{align}
Therefore, we aim to find solutions to \eqref{0eigen}, \eqref{00eigen} such that \eqref{sqint} does not hold. 

Let $A(r) = - \frac{e \phi(r) }{\mEL c^2 f^2(r)} + \frac{1}{f(r)}$,
and use the ansatz $g_2(r) = e^{{h_2}(r)}$. 
 Then, by \eqref{0eigen} we have
\be
g_1=A^{-1}(r) \Big[\frac{\hbar }{\mEL c}\left[{h_2}^\prime - \frac{k}{r f(r)}\right]+\frac{\mu_a\phi^{\prime}(r)}{\mEL c^2 f(r)}\Big] g_2. 
\ee 
 Recall that $f^2(r) \sim 2\frac{G}{c^2}\mNULL r^{-1-\alpha}$ and 
$\phi(r) \sim C^{\prime\prime}_\beta + \phiNULL r^{-\beta}$ around zero with  $ \frac{1+\alpha }{2} < \beta \leq 1$.
 Therefore, plugging $g_1$ into \eqref{00eigen} we obtain the following asympototic expansion as $r \rightarrow 0$,
\begin{align}\hspace{-.5truecm}
 \Big[ \partial_r + (A^{-1} (r))^{\prime} A(r)\Big] 
\Big[ {h_2}^{\prime} - \frac{k}{r  f(r)} +  \frac{\mu_a \phi^{\prime}}{\hbar cf(r)} \Big] + ({h_2}^{\prime}  )^2 -
 \Big[ \frac{k}{r  f(r)} -  \frac{\mu_a \phi^{\prime}}{\hbar c f(r)} \Big]^2  
                                    = O(r^{ 1+\alpha  }).\hspace{-1truecm}
 \end{align} 
 Noting also that $ \beta > \alpha \geq 0$, we can find a solution so that 
\be
 {h_2} (r) \sim C  r^{ - \beta + \frac{1+\alpha}{2}}  \,\,\,\, \mbox{as} \quad r \rightarrow 0^{+} , 
\ee
or equivalently 
\be
  g_2(r)\sim e^{ C r^{ - \beta + \frac{1+\alpha}{2}}} \,\,\,\, \mbox{as}  \quad r \rightarrow 0^{+} .
\ee
In a similar way, one can show that 
\be
 {h_1} (r) \sim -C  r^{ - \beta + \frac{1+\alpha}{2}} \Rightarrow  g_1(r)\sim
 e^{ -C  r^{ - \beta + \frac{1+\alpha}{2}}} \,\,\,\, \mbox{as}\quad r \rightarrow 0^{+}.
\ee
It is now clear that since $\beta >  \frac{1+\alpha}{2} $, \eqref{sqint} does not hold for $g = (g_1, g_2)^T$, 
and the LPC is satisfied in the right neighborhood of zero. 

 Finally, we consider the case $\beta =  \frac{1+\alpha}{2} $. 
Recall that, since $ \alpha \geq 0$,  we have  $f^2(r) \sim 2\frac{G}{c^2} \mNULL r^{-1-\alpha}$ as  $ r \rightarrow 0^{+}$.
 Therefore, equations \eqref{0eigen}, \eqref{00eigen} around zero become
\begin{align} 
& g_2^{\prime} - \frac{ (1+\alpha) \mu_a \phiNULL}{ 2 \sqrt{2 G\mNULL} \hbar r} - \frac{k c} { \sqrt{2 G\mNULL} r^{\frac{1-\alpha}{2}}} = O(r^{1/2}), \\
& g_1^{\prime} + \frac{ (1+\alpha) \mu_a \phiNULL}{ 2  \sqrt{2 G\mNULL} \hbar  r} + \frac{k c} { \sqrt{2 G\mNULL}r^{\frac{1-\alpha}{2}}} = O(r^{1/2}).
\end{align} 
Notice that if $\mu_a \phiNULL >0$ then {in a right neighborhood of $r = 0$ we have}
\be
g_1 \sim r^{ - \frac{(1+ \alpha) \mu_a \phiNULL}{ 2  \sqrt{2 G\mNULL}\hbar}} {;}
\ee
and if  $\mu_a \phiNULL  < 0$ then {in a right neighborhood of $r = 0$ we have}
\be
g_2 \sim r^{\frac{ (1+ \alpha) \mu_a \phiNULL}{ 2  \sqrt{2 G\mNULL}\hbar}}  {.}
\ee 
 Note that  \eqref{sqint} implies that local square integrability holds for $g_1$ and $g_2$ if
\be 
\int_0^{R} r^{ \pm \frac{ (1 + \alpha) \mu_a \phiNULL}{   \sqrt{2 G\mNULL}\hbar} + 1+\alpha } dr < \infty .
\ee

Hence, the LPC is satisfied in the right neighborhood of zero if 
\be
{ - \Big| \frac{ (1+ \alpha) \mu_a \phiNULL}{ \sqrt{2G \mNULL}\hbar} \Big| + 1 + \alpha \leq -1 .}
\ee

\vspace{-1truecm}\end{proof} 

 Our last theorem states that an anomalous magnetic moment can only regularize the Dirac operator for a 
test electron in the static spherically symmetric spacetime of a point nucleus with negative bare mass if
the electric field of the nucleus diverges sufficiently fast at $r\searrow 0$ to overcome the effect of the
spacetime singularity due to the negative bare mass.

Thus, somewhat unexpectedly (to us at least), the Dirac operator for a test electron in the Hoffmann 
spacetime of a point nucleus is essentially self-adjoint \emph{if and only if} the bare mass of the nucleus vanishes. 
 No anomalous magnetic moment can come to the rescue if the bare mass of the nucleus is strictly negative.

 We end this section with the analogues of Theorems~\ref{th:esspec} and \ref{thm:eigenvalues} for test electrons with anomalous magnetic moment.
 By $H_{\mu_a,\theta}$ we denote any self-adjoint extension of $H_{\mu_a}$, where it is understood that the subscript $\theta$ is mute
in all cases where $H_{\mu_a}$ is essentially self-adjoint.

 For the essential spectrum we have:
\begin{thm}\label{th:esspecMU}
 Suppose Assumptions~\ref{massASS} and \ref{potASS} hold, and furthermore assume that 
$ \phi(r)= {C_\beta^{\prime\prime}+}\phiNULL r^{-\beta} + O_1( r^{ \frac{3}{2} - \beta} )$  for $ \beta \leq 1$ around zero. 
 Then one has 
\begin{itemize}
\item[(a)]
  The essential spectrum $\sigma_{ess}(H_{\mu_a,\theta}) = (-\infty,-\mEL c^2]\cup [\mEL c^2, \infty)$;
  \item[(b)] $ H_{\mu_a,\theta}$ has purely absolutely continuous spectrum in $(-\infty,-\mEL c^2 ) \cup (\mEL c^2, \infty)$;
\item[(c)]  the singular continuous spectrum $\sigma_{sc} (H_{\mu_a,\theta})=  \emptyset$.

\end{itemize}
\end{thm}
\begin{proof}
 We use the representation \eqref{Kmu-x} to validate the claims.
 First of all, note that the proof of $(b)$ follows similarly to the proof of Lemma~\ref{acK}. In particular, we need to consider the operator 
\begin{align}
\tilde{P}_2 := P_2(x) + \mu_a d(x) \sigma_1	
\end{align}
instead of $P_2$ in \eqref{def:p2}, and validate the limit property \eqref{lim0} for $\tilde{P}_2$. Here, 
$\sigma_1 = \left[\scriptsize{\begin{array}{cc}  0 &   1 \\   1 & 0	\end{array}}\right]$. 
However, $\lim_{x \rightarrow \infty} d(x)=0$, and hence part $(b)$ of the statement holds. 
It is also clear that the claim of part $(c)$ follows from part $(a)$ and part $(b)$. 
 Therefore it remains to prove part~$(a)$. 

For this part, we need to analyze the operator \eqref{Kmu-x} separately for $ \beta < \frac{\alpha+1}{2}$ and $\beta> \frac{\alpha+1}{2}$.
 If $\beta< \frac{\alpha+1}{2}$, i.e. if the operator \eqref{Kmu-x} has multiple self-adjoint extensions, then the proof of Lemma~\ref{essK} 
is directly applicable.
 In particular, writing 
\begin{align} 
 \widetilde{K}_{\mu_a,k} = \widetilde{K}^0 + \mu_{a}d(x) \sigma_1 + V(x) 
\end{align}
one can show that the functions defined in \eqref{weyl} form a Weyl sequence also for $\tilde{K}^0 + \mu_{a} d(x) \sigma_1$. 
Furthermore, the operator $S$ in \eqref{S} is bounded and therefore, $V$ is  $\widetilde{K}^0 + \mu_{a}d(x) \sigma_1$ compact. 

For the case that $\beta \geq \frac{1 +\alpha}{2}$, we define the operators  $\widetilde{K}^{[0,b]}_{\mu_a,k}$ and
 $\widetilde{K}^{[b, \infty)}_{\mu_a,k}$ as the restriction of $\widetilde{K}_{\mu_a,k}$ to $L^2([0,b])$ and $ L^2([b,\infty])$ respectively. 
Then by Theorem~11.5 in \cite{WeiC}, we have 
\be 
\sigma_{ess} ( \widetilde{K}_{\mu_a,k} ) =   \sigma_{ess} ( \widetilde{K}^{[0,b]}_{\mu_a,k}) \cup  \sigma_{ess} ( \widetilde{K}^{[b, \infty)}_{\mu_a,k}) .
\ee
 Instead of \eqref{weyl} we now use the following Weyl sequence,
\be 
f_{n, \lambda} (x) = \frac{1}{\sqrt{2 n}} e^{- \frac{(x-b)}{2n} + i (x-b) \sqrt{ \lambda^2 - 1}} 
\left[\begin{array}{c} \sqrt{1 + \frac{1}{ \lambda } } \\ i \sqrt{1 - \frac{1}{ \lambda } }
	\end{array}\right];\quad n\in\Nset.
\ee
 Then $\xi(x) =\int_b^{x} [kc(y) + \mu_a d(y)] dy$ for $ b \leq x \leq b+1$ in \eqref{S}, and one can show that 
 $\sigma_{ess} ( \widetilde{K}^{[b, \infty)}_{\mu_a,k})= (-\infty,-1]\cup [1,\infty)$ in a similar way as in the proof of Lemma~\ref{acK}. 

 On the other hand, the operator $\widetilde{K}^{[0,b]}_{\mu_a,k}$ can only have discrete spectrum. 
 To see that, we use Theorem~2 in \cite{HS}. 
In particular, since the limit point case holds, $\widetilde{K}^{[0,b]}_{\mu_a,k}$ has discrete spectrum if also 
\be
 \int_0^{b} |k c(x) + \mu_a d(x)| dx = \infty. 
\ee
 Note that the above statement is true since for $\beta \geq \frac{1 +\alpha}{2}$, $d(x)$ is not locally integrable around zero. 
 \end{proof} 

For the discrete spectrum we have:
\begin{thm} \label{thm:eigenvaluesMU} Let Assumptions~\ref{massASS} and \ref{potASS} be valid.
 Let also $ \phi(r)= {C_\beta^{\prime\prime}+}\phiNULL r^{-\beta} + O_1( r^{ \frac{3}{2} - \beta} )$ for $\beta \leq 1$ around zero.
 Then eigenvalues of any self-adjoint extension $H_{\mu_a,\theta}$ of $H_{\mu_a}$ form a countably infinite set located in the gap of 
the essential spectrum. 
The set of accumulation points of this discrete spectrum $\sigma_{disc}(H_{\mu_a,\theta}^{})$ is either $\{\mEL c^2\}$ or
$\{-\mEL c^2,\mEL c^2\}$,
depending on whether ${GM \mEL}<{Z e^2}$ or ${GM \mEL}>{Z e^2}$, respectively.
 In particular, for {the empirically known} hydrogenic ion parameters, only $\mEL c^2$ is an 
accumulation point of the discrete spectrum if $\beta \neq \frac{\alpha+1}{2}$, or if $\beta = \frac{\alpha+1}{2}$ and 
$\frac{1+\alpha}{2+\alpha}\left| \frac{ \mu_a C^{\prime}_\beta }{ \hbar \sqrt{ 2 G C_\alpha}}\right| \neq 1$. 
\end{thm}
\begin{proof}
For the proof of the fact that the eigenvalues form a countably infinite set located in the gap, we recall the proof of Theorem~\ref{thm:eigenvalues}.
 In particular, we need to apply Theorem~\ref{thm:Hinton} with the same $V_1$ and $V_2$, but $p(x)$ is exchanged with 
\be
 \tilde{p}(x) = - \frac{\hbar}{\mEL c} \frac{k}{r(x)} f(r(x))+ \frac{\mu_a}{\mEL c^2}  \phi^{\prime}(r(x))f(r(x)) = p(x) + \mu_a d(x) .
\ee
 Note that $d(x)$ vanishes faster than  $p(x)$ at infinity. 
 Therefore, the behavior of $\Gamma_{+}$, see \eqref{gammaasymp}, around infinity 
remains the same and, the proof now follows similar to the proof of Theorem~\ref{thm:eigenvalues}. 

 To prove the claim on accumulation points, we follow a similar method as in the proof of Lemma~\ref{lem:accum}.
 In  Lemma~\ref{lem:accum}, note that $1$ is  accumulation point because equation \eqref{zprpr} with $\Gamma^1_+$ has 
oscillatory solutions at $\infty$ or, equivalently, $ \lim_{x \rightarrow \infty} x^2  \Gamma^1_+(x) < - \frac{1}{4}$.
 However, since $\tilde{p}(x) \sim p(x)$ at infinity, the behavior of  $\Gamma^1_+(x)$ remains the same and, $ \mEL c^2 $ is an accumulation point. 

 Next, we prove the statement about $ - \mEL c^2 $. To do that, we have to consider the equation \eqref{zprpr} 
with $\Gamma^{-1}_{\pm}$, where  $p(x)$ is replaced by $\tilde{p}(x)$; see \eqref{gamma+-1}, \eqref{gamma--1}. 
  Again since $\tilde{p}(x) \sim p(x)$ at infinity, the solutions to \eqref{zprpr} with $\Gamma^{-1}_{\pm}$ are 
non-oscillatory if ${GM\mEL}<{Ze^2}$, and oscillatory if ${GM\mEL}>{Ze^2}$.
Thus, when ${GM \mEL}>{Z e^2}$, then $ -\mEL c^2 $ is an accumulation point.

 Now we need to determine if the solutions to \eqref{zprpr} with $\Gamma^{-1}_\pm$, and $p(x)$ replaced by $\tilde{p}(x)$, 
are non-oscillatory also around zero when  ${GM\mEL}<{Ze^2}$.
 This part of the proof requires more care since the behavior of $\Gamma^{-1}_\pm$ around zero is affected when $p(x)$ is replaced by $\tilde{p}(x)$.  
Note that if $\beta \leq 0$, then $p(x)$ is more singular than $d(x)$ as $x \rightarrow 0$. 
Therefore,  Corollary~3.14 in \cite{HMRS} is applicable as in the proof of Theorem~\ref{thm:eigenvalues} if $\beta \leq 0$.  
On the other hand if $\beta > 0$, then the most singular term in both $\Gamma^{-1}_\pm$ arises from $\tilde{p}^2$ with singularity 
$x^{-\frac{4 \beta  +2 \alpha +6}{ 4 + 2 \alpha}}$ if $ \beta > \frac{1+\alpha}{2}$;
 and from $\tilde{p}^{\prime}$ with singularity $x^{-\frac{2 \beta + 3 \alpha +7}{ 4 + 2 \alpha}}$ if $ \beta < \frac{1+\alpha}{2}$. 
In particular, if $ \beta > \frac{1+\alpha}{2}$ the singularity comes from the term $(\phi^{\prime} f )^2$, and if $ \beta < \frac{1+\alpha}{2}$
then the singularity comes from $(\phi^{\prime} f )^{\prime}$.
 Noting that $p^2$ has $+$ sign in both $\Gamma^{-1}_\pm$, we see that if
 $ \beta > \frac{1+\alpha}{2}$ then $ \lim_{x \rightarrow 0} x^2 \Gamma^{-1}_\pm \geq 0$.
  On the other hand,  if $ \beta < \frac{1+\alpha}{2}$ then $ {\frac{2 \beta + 3 \alpha +7}{ 4 + 2 \alpha}} <2$,
 and hence $ \lim_{x \rightarrow 0} x^2 \Gamma^{-1}_\pm =0$.
 Therefore, the solutions are non-oscillatory if  $ \beta \neq \frac{1 + \alpha}{2} $.

  To determine, the behavior in the case of the equality we need to track the exact coefficient of the term $x^{-2}$.
 We determine this coefficient as 
\begin{align}
\Big(\frac{ (1+\alpha)\mu_a  C^{\prime}_\beta }{ 2(2 + \alpha)\hbar \sqrt{ 2 G C_\alpha}} \Big)^2 \pm
 \frac{ (1+\alpha)\mu_a C^{\prime}_\beta }{ 2(2 + \alpha)\hbar \sqrt{ 2 G C_\alpha}}
\end{align} 
Hence, $ \lim_{x \rightarrow 0} x^2 \Gamma^{-1}_- > -1/4 $ as long as 
 $\Big| \frac{1+\alpha}{2 + \alpha}\frac{\mu_a C^{\prime}_\beta }{\hbar \sqrt{ 2 G C_\alpha}}\Big| \neq 1$. 

 So for $\beta \neq \frac{1+\alpha}{2} $ then $ -\mEL c^2 $ is not an accumulation point when ${GM\mEL}<{Z e^2}$. 
 On the other hand, if $ \beta = \frac{1+\alpha}{2}$ then $ -\mEL c^2$ is not accumulation point if 
$\frac{1+\alpha}{2+\alpha}\Big| \frac{\mu_aC^{\prime}_\beta }{\hbar \sqrt{ 2 G C_\alpha}}\Big| \neq 1$ holds together with ${GM\mEL}<{Ze^2}$. 
\end{proof}

\section{Summary and outlook}\label{sec:SUMM}

 We have discussed the Dirac operator for a test electron in the static spherically symmetric
spacetime of a point nucleus with negative bare mass, allowing for a large class of electromagnetic vacuum laws 
compatible with the form of the spacetime metric given in (\ref{dsSQR}). 
 We have considered test electrons without and with an anomalous magnetic moment. 
Our findings demonstrate that the theory of the Dirac operator of a test electron in
even this simple class of spherically symmetric electrostatic spacetimes is rich and full of surprises!

 Different from the essentially self-adjoint situation which prevails when the bare mass of the nucleus vanishes, 
which was considered in \cite{Moulik}, the Dirac operator is never essentially self-adjoint when the bare mass of the
nucleus is strictly negative --- unless the test electron features an anomalous magnetic moment. 
 Even then, essential self-adjointness holds only if the electric field of the nucleus diverges sufficiently fast
at the nucleus, which is not the case for a large subset {of the} electromagnetic vacuum laws considered. 
 In particular, it is not the case for the Born--Infeld vacuum law. 
 Furthermore, on spacetimes of nuclei with (possibly infinite) negative bare mass and sufficiently rapidly diverging 
electric field, if the electric field diverges precisely at the critical rate then the anomalous magnetic
moment has to be sufficiently strong to guarantee essential self-adjointness. 
 In the special case of the Reissner--Weyl--Nordstr\"om spacetime of a point nucleus our formula for the 
critical value of the electron's anomalous magnetic moment coincides with the one found previously in \cite{BMB}.

 For all self-adjoint extensions of our Dirac operators we identified the essential spectrum with the usual gap
{$(-\mEL c^2, \mEL c^2)$} and showed that the gap contains infinitely many eigenvalues.
 When $\frac{GM\mEL}{Ze^2}<1$ there is generally only one family of eigenvalues, with $\mEL c^2$ as cluster point. 
 Yet when $\frac{GM\mEL}{Ze^2}>1$ there are two families of eigenvalues, one with $\mEL c^2$ as cluster point,
and another one with $-\mEL c^2$ as cluster point. 

 However, the hyper-heavy nucleus regime $\frac{GM\mEL}{Ze^2}>1$ is not realized in nature if
the empirical formula for the nuclear masses, $M=A(Z,N)\mPR$ with $A\approx Z+N$ and $N\leq 2Z$,
continues to hold for arbitrary $Z$ and $N$.
 Namely, the hyper-heavy nucleus condition $\frac{GM\mEL}{Ze^2}>1$ implies $\frac{GM Z\mEL}{Z^2e^2}>1$, 
and since $\mPR\approx 1836\mEL$ and $M=A\mPR$ with $A\geq Z$, this implies $\frac{GM^2}{Z^2e^2}>1$, which means
we are in the black hole sector. 
 But this is impossible if also $N\leq 2Z$ continues to hold for arbitrary $N$ and $Z$,
because the empirical charge to mass ratio of the proton together with $M=A(Z,N)\mPR$ and $A\approx Z+N$  and 
$N\leq 2Z$ implies that  $\frac{GM^2}{Z^2e^2}<1$ ($\ll 1$ in fact).
 Hence we have a contradiction.

 If one drops the assumption that $N\leq 2Z$ (as in a neutron star), then the hyper-heavy nucleus condition can be made 
compatible with the black hole sector condition and the mass formula  $M=A(Z,N)\mPR$ with $A\approx Z+N$.
 However, our results do not apply to the black hole sector, and it is an interesting 
open question whether the Dirac Hamiltonian acting on bi-spinor wave functions of a
test electron supported entirely inside the Cauchy horizon of the black hole spacetime of a hyper-heavy nucleus
with mass formula  $M=A(Z,N)\mPR$ with $A\approx Z+N$ is well defined (with or without anomalous magnetic moment 
taken into account), or at least has self-adjoint extensions, and if so, whether there are two families of eigenvalues with 
cluster points $\pm\mEL c^2$. 
 While this may never be of concern to experimental physicists, for the satisfaction of intellectual curiosity we have
begun to investigate this problem \cite{KKTtwo}.

 In any case the mathematical spectra of hypothetical hyper-heavy \emph{naked} nuclei are not realized in nature according
to our analysis.

 Only if one drops the mass formula  $M=A(Z,N)\mPR$ with $A\approx Z+N$ completely and treats $M$, $\mEL$ and $e$ as parameters, 
then the family of `hyper-heavy hydrogenic ion' eigenvalues having cluster point $-\mEL c^2$ can exist on a naked singularity spacetime, 
mathematically speaking, but it would be a bit of a stretch to refer to it as a hyper-heavy hydrogenic ion. 
 Whether there is any physical scenario which could lead to such a situation in nature, or whether this is pure science fiction, we don't 
know, but it may be worth pondering.
 
 Beside the hyper-heavy hydrogenic `black hole ion' problem mentioned above, there are a number of spectral questions which we have not 
answered, such as whether the boundary points $\pm \mEL c^2$ of the essential spectrum are eigenvalues.
 We have also not attempted to determine the discrete spectra in detail, which is worth the effort only if
one has a compelling candidate for the physically correct self-adjoint $H$.

 Then there are electromagnetic vacuum laws such as the one proposed by {Bopp, Land\'e--Thomas, and
Podolsky} (BLTP) which are not compatible with the form of the spacetime metric given in (\ref{dsSQR}). 
 A similar study such as the one conducted in this paper should also be carried out for vacuum laws of the BLTP type.

 The test electron approximation can be expected to be very accurate for large $Z$ but certainly less so for
hydrogen ($Z=1$).
 Therefore it is desirable to overcome the test electron approximation. 
 This has so far only been accomplished in a fully satisfactory manner in the non-relativistic Schr\"odinger model of 
hydrogenic ions.
 We consider it to be one of the most challenging and important open problems of rigorous relativistic quantum mechanics.

\smallskip
\noindent{\textbf{Acknowledgement}: We thank Moulik Balasubramanian for interesting discussions.}

 \newpage

\begin{appendix}
\numberwithin{equation}{section}
\section{}\vspace{-5pt}
 In this appendix we show that general relativity has a regularizing effect on the 
 Bohr--Sommerfeld-type model of hydrogenic large-$Z$ ions with Coulomb interactions.
We also demonstrate this effect when Coulomb interactions are replaced by electric interactions
in a nonlinear electrostatic vacuum.
 Like Bohr we work for simplicity only with circular orbits.\vspace{-5pt}
\subsection{Coulomb interactions}
 {Following} Vallarta, we here assume that the static spacetime of a point nucleus is given by the
Reissner--Weyl--Nordstr\"om solution of the Einstein--Maxwell system.
 Then the general-relativistic Bohr--Sommerfeld-type energies $E_n^{GR}(Z,N), n\in\Nset$, of a hydrogenic ion
with a nucleus of charge $Ze$ and mass $A(Z,N)\mPR$, with $Z\leq A(Z,N) < 3Z$ for the known nuclei,
is determined by finding, for each $n\in\Nset$, the minimum w.r.t. $r$ of 
\begin{equation}
U_n^{GR}(r) := 
\mEL c^2\sqrt{1+\frac{n^2\hbar^2}{\mEL^2c^2r^2}}\sqrt{1-\frac{2G}{c^4r}\left[A(Z,N)\mPR c^2-\frac{Z^2e^2}{2r}\right]} 
- \frac{Ze^2}{r} .
\label{eq:HnrGRELz}
\end{equation}
 Switching to the dimensionless variables $\rho = r \mEL c /\hbar$ and $V_n^{GR}(\rho) = U_n^{GR}(r)/\mEL c^2$ yields
\begin{equation}
V_n^{GR}(\rho) := 
\sqrt{1+n^2\tfrac{1}{\rho^2}\Big.}\,\sqrt{1-\alphaS\gamma_{pe}\left[2A(Z,N)-\eps\alphaS Z^2\tfrac{1}{\rho}\right]\tfrac{1}{\rho}}-Z\alphaS\tfrac{1}{\rho}.
\label{eq:HnrGRELdimLESSz}
\end{equation}
 Recall that $\alphaS:={e^2}/{\hbar c}\approx 1/137.036$, that $\eps:=\mEL/\mPR \approx 1/1836$, 
and that $\gamma_{pe}:={G\mEL\mPR}/{e^2}\approx 4.5\times 10^{-40}$. 
 Asymptotically for very large $\rho$ the dimensionless general-relativistic energy function 
$V_n^{GR}(\rho)\sim 1 -  {\alphaS(Z + A(Z,N)\gamma_{pe})}/{\rho}$, 
just like the special-relativistic one with Newtonian gravity added to the Coulombian electricity.
 As $\rho$ becomes smaller and smaller, the term $\left[1-\eps\alphaS\frac{Z^2}{2A(Z,N)}\frac{1}{\rho}\right]$ changes sign,
namely for  $\rho < \frac12\frac{Z^2}{A(Z,N)}\eps\alphaS$ it is negative. 
 We see that `overall' the zero  $\rho_0(Z,N) :=\frac12\frac{Z^2}{A(Z,N)}\eps\alphaS$ grows with $Z$; more precisely, 
$\frac16 \eps\alphaS Z \leq \rho_0(Z,N) \leq \frac12 \eps\alphaS Z$. 
 Now, the smallest such $\rho(Z,N)$ where the sign switch happens is a tiny dimensionless distance, and the factor $-1/\rho$ before
the $[\quad]$-bracketed term could threaten that the whole expression under the square root becomes negative before this tiny $\rho$
is reached (starting from large $\rho$ and making $\rho$ smaller and smaller). 
 Yet, since $\gamma_{pe}$ is much tinier yet, the whole expression under the square root remains positive for all $\rho$. 

 For very small $\rho$ the general-relativistic gravitational square-root factor in (\ref{eq:HnrGRELdimLESSz}) contributes a 
factor $Z \sqrt{\gamma_{pe} \eps } \alphaS/\rho$ while the special-relativistic square-root factor in (\ref{eq:HnrGRELdimLESSz})
contributes a factor $n/\rho$ to the total square-root term.
 So for very small $\rho$ the asymptotic behavior is 
$V_n^{GR}(\rho)\sim  Z n\sqrt{\gamma_{pe} \eps}\alphaS/{\rho^2} \nearrow \infty$ as $\rho\searrow 0$, and except for the different
coefficient, this is like the behavior of the non-relativistic kinetic energy function ($\propto n^2/\rho^2$) 
 {in Bohr's model with
circular orbits.
 Thus $V_n^{GR}(\rho)$ always has a minimum at a strictly positive $\rho$ for each $Z \in\Nset$ and $n\in\Nset$.}
\newpage

 The special-relativistic version of this problem is qualitatively very different.
 Setting $G\searrow 0$ in (\ref{eq:HnrGRELz}) yields
\begin{equation}
U_n^{SR}(r) := 
\mEL c^2\sqrt{1+\frac{n^2\hbar^2}{\mEL^2c^2r^2}} - \frac{Ze^2}{r} ,
\label{eq:HnrSRELz}
\end{equation}
which is the same as setting $\gamma_{pe}\searrow 0$ in (\ref{eq:HnrGRELdimLESSz}), viz.
\begin{equation}
V_n^{SR}(\rho) := 
\sqrt{1+n^2\tfrac{1}{\rho^2}} - Z\alphaS\tfrac{1}{\rho} .
\label{eq:HnrSRELdimLESSz}
\end{equation} 
 Finding for each $n\in\Nset$ the minimum w.r.t. $r$, respectively $\rho$, will produce the principal energy values of Sommerfeld's
fine structure spectrum of a hydrogenic ion whenever $Z\leq 137$, but for each $n$ the bottom drops out when $Z> Z_*(n) 
{\geq Z_*(1)}$, with $Z_*(1) = 137$.\vspace{-5pt}

\subsection{Nonlinear electrostatic vacuum}

 It follows from the discussion in the main text that 
replacing Maxwell's ``law of the pure ether'' by a nonlinear electromagnetic vacuum law of the type 
considered in \cite{TZ}, and in this paper, amounts to replacing $\frac12{Z^2e^2}/{r}$ by $\cE(r)$ in 
(\ref{eq:HnrGRELz}) and ${Ze^2}/{r}$ by $e\phi(r)$ in (\ref{eq:HnrGRELz}) and in (\ref{eq:HnrSRELz}).
 The class of nonlinear vacuum laws considered in this paper weakens the Coulomb singularity to 
$\phi(r) \sim C_\beta^{\prime\prime} + C_\beta^{\prime}r^{-\beta}$ as $r\searrow 0$, with $\beta<1$, 
and this already removes the `large $Z$ catastrophe' from the corresponding special-relativistic Bohr--Sommerfeld type theory.
 The question thus becomes whether the general-relativistic square-root factor in (\ref{eq:HnrGRELz}) can now cause a 
spectral catastrophe, or not. 

 Since we consider only black hole-free spacetimes of nuclei, we have
$\lim_{r\searrow 0} m(r) = A(Z,N)\mPR  - \frac{1}{c^2}\cE(0)\leq 0$; cf. section 2.
 We need to distinguish $ m(0) = 0$ and $m(0) < 0$.

 Suppose first that $\lim_{r\searrow 0} m(r) < 0$ (possibly $-\infty$). 
 Then the general-relativistic square-root factor in (\ref{eq:HnrGRELz}) diverges $\propto {1}/{r^\kappa}$
as $r\searrow 0$, which together with the $1/r$ singularity of the special-relativistic square-root factor in (\ref{eq:HnrGRELz}) 
yields an overall ${1}/{r^{1+\kappa}}$ singularity, with $\kappa \geq \frac12$ (N.B. $\kappa=\frac12$ if 
$\lim_{r\searrow 0} m(r) = m(0) < 0$ exists, and $\kappa \in (\frac12,1]$ if 
$\lim_{r\searrow 0} m(r) = -\infty$). 
 This overpowers the $r^{-\beta}$ singularity with $0<\beta<1$.
 There is no `large-$Z$ catastrophe'.

 Consider next the case where $\lim_{r\searrow 0} m(r) = 0$ \emph{for all $Z$}. (Admittedly this is presumably a
purely academic situation, but it's feasible mathematically.)
 In this case we need also an assumption about how $m(0)=0$ is approached. 
 We consider power laws $m(r) = A r^\varkappa$ with $\varkappa >0$ and $A>0$.
 Then as $r\searrow 0$ the general-relativistic square-root factor in (\ref{eq:HnrGRELz}) diverges $\propto {1}/{r^{(1-\kappa)/2}}$
if $\varkappa\in(0,1)$, and otherwise converges to a positive constant $\leq 1$ if $\varkappa \geq 1$, with ``$<$'' iff
$\varkappa =1$.
 Together with the $1/r$ singularity of the special-relativistic square-root factor in (\ref{eq:HnrGRELz}) this
yields a ${r^{({\min\{\varkappa,1\}-3})/{2}}}$ singularity which once again
overpowers the $1/r^\beta$ singularity with $0<\beta<1$.
 There is no `large-$Z$ catastrophe' in this case either.

\section{The family of electrovacuum spacetimes}
\label{app:electrovac}
 In this appendix we present a large family of static spherically symmetric electromagnetic vacuum spacetimes, one that 
includes both the RWN as well as Hoffmann's with either zero or negative bare mass, and all the members of which satisfy 
the assumptions \ref{massASS} and \ref{potASS} made in Section 3. 
 We begin by recalling \cite{TZ} that all such spacetimes are characterized by the choice of a single $C^2$ function of 
one variable $\zeta: \RR_+ \to \RR_+$, called the {\em reduced electromagnetic Hamiltonian} that satisfies the following properties
\begin{itemize}
\item[\textbf{(R1)}] $\lim_{\mu \to 0} {\zeta(\mu)}/{\mu} = 1$,
\item[\textbf{(R2)}] $\zeta'>0$ and $\zeta(\mu) - \mu \zeta'(\mu) \geq 0 \qquad \forall\mu>0$.
\item[\textbf{(R3)}] $\zeta'(\mu) + 2 \mu \zeta''(\mu) \geq 0\qquad \forall \mu>0$.
\end{itemize}
The first condition ensures that the electromagnetic vacuum law agrees with Maxwell's in the weak field limit. 
 The second condition guarantees that the energy tensor of the theory satisfies the dominant energy condition, and the third 
condition is equivalent to this theory being derivable from an action principle with a single-valued Lagrangian. 
The reduced Hamiltonian corresponding to Maxwell's vacuum law is $\zeta(\mu) = \mu$, while the one corresponding 
to {Born's law is $\zeta(\mu) = \sqrt{1+2\mu} - 1$;
in the electrostatic special case this coincides with the one from Born--Infeld's vacuum law}.

It was shown in \cite{TZ} that for every such choice of $\zeta$, and {parameters} $M>0, Q\in \RR$, 
there is a corresponding static, 
spherically symmetric, asymptotically flat solution of {the} Einstein--Maxwell system with ADM mass $M_{\mbox{\tiny{ADM}}}$ 
equal to $M$, and total charge $Q$, 
as described in Sec.~\ref{sec:EST}, with metric line element (\ref{dsSQR}) defined in terms of the radial mass function $m(r)$ 
as in (\ref{mOFr}), where
\begin{equation}\label{def:massfunc}
m(r) = M - \frac{1}{c^2} \int_r^\infty \zeta\Big(\tfrac{Q^2}{2s^4}\Big)s^2 ds
\end{equation}
and the electrostatic potential
\begin{equation}\label{def:elecpot}
\phi(r) = Q \int_r^\infty \zeta'\Big(\tfrac{Q^2}{2s^4}\Big)\tfrac{1}{s^2} ds.
\end{equation}
 It was further shown in \cite{TZ} that under additional assumptions on $\zeta$, one could make sure that the singularity present 
at the center of these spacetimes (which is not shielded by a horizon) is of the mildest form possible, namely a conical singularity, 
with zero bare mass $m(0) = 0$, and that for these spacetimes the ADM mass $M_{\mbox{\tiny{ADM}}} = m(\infty)$ 
is equal to the total electrostatic energy.
  The prime example of these is the Hoffmann spacetime discussed {in section~\ref{sec:EST}}.

 Here we introduce a larger class of such spacetimes that includes, in addition to mildly singular manifolds like 
Hoffmann's, also those with much more severe singularity at the center, such as the RWN, which has negative infinite bare mass. 
 We begin by introducing a one-parameter family of reduced Hamiltonians $\zeta_1$, parametrized by a positive number $\mu_1>0$, 
as follows:
\begin{equation}
\zeta_1(\mu) := \min\{ \mu, \sqrt{\mu_1 \mu} \}.
\end{equation}
We note that $\zeta_1$ is Lipschitz continuous and satisfies assumptions {\bf (R1--R3)} away from its kink at $\mu=\mu_1$.  
The {mass function corresponding} to $\zeta_1$ is denoted  by $m_1(r)$.  It is a $C^1$ function, and can be computed from 
(\ref{def:massfunc}):
\begin{equation}
m_1(r) := \left\{\begin{array}{ll}
M - \frac{Q^{3/2}}{2^{3/4}c^2} \mu^{1/4} & \mbox{if}\quad \mu < \mu_1 ,\\
\frac{M^2 c^2}{2^{5/4} Q^{3/2}} \mu^{-1/4} & \mbox{if}\quad \mu > \mu_1;\end{array}\right.\quad \mbox{with}\quad \mu(r) := \frac{Q^2}{2r^4}.
\end{equation}
Next we show that for a particular choice of the parameter $\mu_1$ this becomes a model for the vacuum spacetime outside a 
point charge of mass $M=M_{\mbox{\tiny{ADM}}}$ and charge $Q$:  Let 
\begin{equation} r_1 := \frac{Q^2}{Mc^2}
\end{equation}
denote the ``classical radius" of this point charge (i.e. the distance at which its electrostatic self-energy equals the rest 
energy of the particle), and set $\mu_1 = \mu(r_1)$. We then obtain
\begin{equation}
m_1(r) 
:= 
\left\{\begin{array}{ll}
     \frac{M^2 c^2}{2 Q^2} r &\mbox{for}\quad  r < r_1, \\
      M - \frac{Q^2}{2c^2 r} & \mbox{for}\quad r > r_1.
         \end{array}
\right.
\end{equation}
It is easy to verify that the above mass function satisfies the assumptions \ref{massASS} that were made in Section 3, provided 
that the mass $M$ and charge $Q$ of the particle satisfy the ``no horizon" condition
\begin{equation}\label{noBH} \frac{GM^2}{Q^2} < 1. \end{equation}

We are now ready to define a whole family of electrovacuum spacetimes with mass functions and electrostatic potentials that satisfy 
the assumptions \ref{massASS} and \ref{potASS} and can serve as models for the vacuum outside a point charge of mass 
$M_{\mbox{\tiny{ADM}}}$ and charge $Q$ and arbitrary finite or infinite negative bare mass:

\begin{prop} Let $M\in \RR_+$ and $Q\in \RR$ be given, subject to (\ref{noBH}).
  Let $\zeta: \RR_+ \to \RR$ be any $C^2$ function 
satisfying assumptions {\bf (R1--R3)} above, and in addition assume
\begin{equation}\label{zetaCOND}
 \zeta(\mu) > \min\{ \mu, \sqrt{\mu_1\mu}\},\qquad \mu_1 := \frac{M^4 c^8}{2 Q^6}.
\end{equation}
Then the corresponding static, spherically symmetric, asymptotically flat solution of the
Einstein--Maxwell equations, with vacuum law 
given by $\zeta$, is characterized by the mass function $m(r)$ as in (\ref{def:massfunc}) and electrostatic potential $\phi(r)$ 
as in (\ref{def:elecpot}) that satisfy the assumptions \ref{massASS} and \ref{potASS} made in Section 3 of this paper.
\end{prop}

\begin{proof}
We first verify assumptions \ref{massASS}.  From (\ref{def:massfunc}) $m$ is clearly a $C^1$ function of $r$.  Moreover
\begin{equation}
m(r) = M - \frac{1}{c^2}\int_r^\infty \zeta(\mu(s)) s^2 ds  \leq  M - \frac{1}{c^2} \int_r^\infty \zeta_1(\mu(s)) s^2 ds = m_1(r),
\end{equation}
so that the mass function of this manifold sits below the mass function $m_1(r)$ of the model spacetime we constructed in the above, 
and hence satisfies the no-black-hole condition $\frac{m(r)}{r} \leq \frac{ c^2}{2G}$, since $m_1(r)$ is seen to satisfy this condition.

We note that $m(r) \leq m_1(r)$ allows these spacetimes to have negative bare mass that can be finite or infinite. 

Consider now the function $\zeta$.  It is a smooth and by {\bf (R1,R2)} positive and increasing function of its argument.  By 
integrating the differential inequalities in {\bf (R2,R3)} one obtains that
\begin{equation}\label{zetabounds}
\forall\ 0<\mu'<\mu:\quad \frac{\mu'}{\mu} \leq \frac{\zeta(\mu')}{\zeta(\mu)} \leq \sqrt{\frac{\mu'}{\mu}}.
\end{equation}
Suppose now that $a,b\in\RR$ and $c_a,c_b>0$ are such that
\begin{equation} \label{zetaASS}
\zeta(\mu) 
\sim 
\left\{
\begin{array}{ll} 
c_a\mu^a & \mbox{for}\quad \mu \to 0,\\
c_b\mu^b & \mbox{for}\quad \mu \to \infty.
\end{array}
\right.
\end{equation}
Then by (\ref{zetabounds}) and {\bf (R1)} we have that $a = 1$, $c_a = 1$, and $ \frac{1}{2}\leq b \leq 1$.   

Next we observe that
\beq\label{mATnull}
m(0) = M - \frac{1}{c^2}\int_0^\infty \zeta(\mu(s)) s^2 ds.
\eeq
Consider the following two cases: (F) The integral in (\ref{mATnull}) is finite, and (I) that integral is infinite.  

If that integral is finite, then by the additional assumption we have made about $\zeta$, namely (\ref{zetaCOND}), we have
\beq
m_0 := m(0) \leq m_1(0) = 0.
\eeq
Moreover, since
\beq
\int_0^\infty \zeta(\mu(s)) s^2 ds = \frac{Q^{3/2}}{2^{11/4}}\int_0^\infty \mu'^{-7/4}\zeta(\mu') d\mu',
\eeq
it is clear that Case (F) corresponds to $b<\frac{3}{4}$ and Case (I) to $b \geq \frac{3}{4}$.  

In Case (F), we can express the mass function in an alternative way,
\beq
m(r) = m_0 + \frac{1}{c^2} \int_0^r \zeta(\mu(s)) s^2 ds, 
\eeq
where $m_0 := m(0) \leq 0$ is the {\em bare mass} of the central singularity.  
 The asymptotics we have established for $\zeta$ then imply that, with $ \lambda := 3-4b$, 
\beq\label{posla}
m(r) \sim \left\{\begin{array}{ll} m_0 + A_\lambda r^\lambda \, &\mbox{as}\  r \to 0 ,\\ 
M - \frac{Q^2}{2c^2r} \, &\mbox{as}\  r \to \infty,\end{array}\right.
\qquad \lambda >0.
\eeq
In Case (I), on the other hand, we obtain
\beq\label{negla}
m(r) \sim \left\{\begin{array}{ll} B_\lambda r^{\lambda} \, &\mbox{as}\  r \to 0 ,\\
M - \frac{Q^2}{2c^2r} \, &\mbox{as}\  r \to \infty,\end{array}\right.\qquad \lambda <0.
\eeq
(The borderline case $b=3/4$ is more subtle due to logarithmic divergence.  We will not consider it here.)
{ We note that (\ref{posla}) corresponds to $\alpha=0$ in Assumptions \ref{massASS} while 
(\ref{negla}) corresponds to $\alpha > 0$ in Assumptions \ref{massASS}, and $\lambda=-\alpha$ then.}

Having established that Assumptions \ref{massASS} 
are satisfied for this family of spacetimes, we move on to the analysis of the 
electrostatic potential $\phi$.  First, we observe that by {\bf (R2,R3)} and (\ref{zetabounds}), 
\beq
\zeta'(\mu_0)\sqrt{\frac{\mu_0}{\mu}} \leq \zeta'(\mu) \leq 1,\qquad\mbox{ for all } 0<\mu_0<\mu.
\eeq
It follows that $\zeta'$ inherits the following asymptotics from $\zeta$:
\beq
\zeta'(\mu) \sim \left\{\begin{array}{ll} 1 \, &\mbox{as}\  \mu \to 0,\\ 
bc_b\mu^{b-1} \, &\mbox{as}\  \mu \to \infty,\end{array}\right.
\eeq
with $\frac{1}{2}\leq b \leq 1$ and hence, in Case (F) the potential $\phi$ satisfies, once again with $\lambda := 3-4b$,
\beq\label{posLA}
\phi(r) \sim \left\{\begin{array}{ll} \widetilde{C}''_\lambda - \widetilde{C}'_\lambda r^\lambda \, &\mbox{as}\  r \to 0,\\
\frac{Q}{r} \, &\mbox{as}\  r \to \infty, \end{array}\right. \qquad \lambda > 0,
\eeq
{with $\widetilde{C}'_\lambda >0$ and $\widetilde{C}''_\lambda >0$ if $Q>0$}, 
while in Case (I) we have 
\beq\label{negLA}
\phi(r) \sim \left\{\begin{array}{ll} \widehat{C}'_\lambda r^{\lambda} \, &\mbox{as}\  r \to 0,\\
\frac{Q}{r} \, &\mbox{as}\  r \to \infty. \end{array}\right. \qquad \lambda <0.
\eeq
{with $\widehat{C}'_\lambda >0$ if $Q>0$.
 We note that (\ref{posLA}) corresponds to $\beta\leq 0$ in Assumptions \ref{potASS} while 
(\ref{negLA}) corresponds to $\beta>0$ in Assumptions \ref{potASS}, and $\lambda=-\beta$ then.
 Note that $\lambda\geq -1$.}
\end{proof}

\end{appendix}

\newpage
\centerline{\textbf{Data availability statement}:}

Data sharing is not applicable to this article as no new data were created or analyzed in this study.

\vfill\vfill
\hrule
\end{document}